\documentclass[a4paper,11pt]{article} 
\usepackage[margin=1in]{geometry}

\usepackage{graphicx}
\usepackage{subcaption}
\usepackage{amsthm}
\makeatletter
\def\th@plain{%
	\thm@notefont{}
	\itshape 
}
\def\th@definition{%
	\thm@notefont{}
	\normalfont 
}
\makeatother
\usepackage{amssymb}
\usepackage{amsmath}
\usepackage{mathrsfs}
\usepackage{verbatim}
\usepackage{algorithm}
\usepackage{mathtools}
\usepackage{xspace}
\usepackage{enumitem}
\usepackage{color}
\usepackage{url}
\usepackage[all]{xy}
\usepackage{microtype}
\usepackage{array}
\usepackage{float}
\usepackage{algorithmic}
\usepackage{lmodern}
\usepackage{anyfontsize}
\usepackage{cleveref}
\usepackage[normalem]{ulem}
\usepackage{epigraph}
\usepackage{wasysym}
\usepackage{lineno}

\pdfoutput=1

\newcommand{\R}{\mathbb{R}}
\newcommand{\Z}{\mathbb{Z}}
\newcommand{\B}{\mathcal{B}}

\newcommand{\V}{\mathbb{V}}
\newcommand{\W}{\mathbb{W}}

\newcommand{\eps}{\varepsilon}
\newcommand{\Cech}{\v{C}ech\xspace}
\newcommand{\cech}{\mathcal{C}}

\newcommand{\lazy}{\mathcal{LB}}

\newcommand{\lpixels}{\mathcal{LS}}
\newcommand{\lcomplex}{\mathcal{LX}}

\newcommand{\pixels}{\mathcal{S}} 
\newcommand{\complex}{\mathcal{X}}

\newcommand{\field}{\mathcal{F}}

\newcommand{\nrv}{\mathrm{nerve}}
\newcommand{\nerve}{\mathrm{nerve}}

\newcommand{\simplicialmap}{\varphi}

\newcommand{\linearmap}{\lambda}
\newcommand{\id}{\mathrm{id}}
\newcommand{\inc}{\mathrm{inc}}

\newcommand{\distance}[2]{\|#1-#2\|}

\newcommand{\ignore}[1]{}
\newcommand{\homt}{\overset{h}{\simeq}}

\newtheorem{lemma}{Lemma}
\numberwithin{lemma}{section}

\newtheorem{theorem}[lemma]{Theorem}

\title{\Large Improved Topological Approximations by Digitization}
\author{
	Aruni Choudhary\thanks{Freie Universit\"at Berlin, Berlin, Germany}
	\and 
	Michael Kerber\thanks{Technische Universit\"at Graz, Graz, Austria} 
	\and
	Sharath Raghvendra\thanks{Virginia Tech, Blacksburg, USA}
}
\date{}

\begin{document}
\maketitle

\begin{abstract}
\Cech complexes are useful simplicial complexes for computing
and analyzing topological features of data that lies in Euclidean space.        
Unfortunately, computing these complexes becomes prohibitively
expensive for large-sized data sets even for medium-to-low dimensional data.
We present an approximation scheme for $(1+\eps)$-approximating the topological 
information of the \Cech complexes for $n$ points in $\R^d$, for $\eps\in(0,1]$.
Our approximation has a total size of $n\left(\frac{1}{\eps}\right)^{O(d)}$
for constant dimension $d$, 
improving all the currently available $(1+\eps)$-approximation schemes
of simplicial filtrations in Euclidean space.
Perhaps counter-intuitively, we arrive at our result by adding additional 
$n\left(\frac{1}{\eps}\right)^{O(d)}$ sample points to the input.
We achieve a bound that is independent of the spread of the point set
by pre-identifying the scales at which the \Cech complexes changes 
and sampling accordingly.
\end{abstract}

\section{Introduction}
\label{section:introduction}

\paragraph{Context and Motivation.}
\emph{Topological data analysis} attempts to extract relevant information
out of a data set by interpreting data as a shape and understanding the
connectivity of this shape.
Connectivity clearly includes the study of connected components, but also
the presence of tunnels, voids, and other high-dimensional topological
features.

How do we interpret data as a shape? We consider the common scenario
that our data is a set of $n$ points $P$ in a possibly high-dimensional Euclidean
space $\R^d$. 
Fixing a positive real scale parameter $\alpha$,
we define $\B_\alpha(P)$ (or simply $\B_\alpha$) as the union of balls of radius $\alpha$
centered at each point in $P$.
The shape $\B_\alpha$ can be seen as an approximation of the shape formed by $P$ for a parameter $\alpha$. 
As seen in Figure~\ref{figure:union_of_balls}, for different values of $\alpha$, the union $\B_\alpha$
may form different topological configurations raising the question of the ``right'' $\alpha$-value
that accurately captures the shape of $P$.

\begin{figure}[ht!]
\centering
\begin{subfigure}[t]{0.3\columnwidth}
\centering
\includegraphics[width=0.65\textwidth]{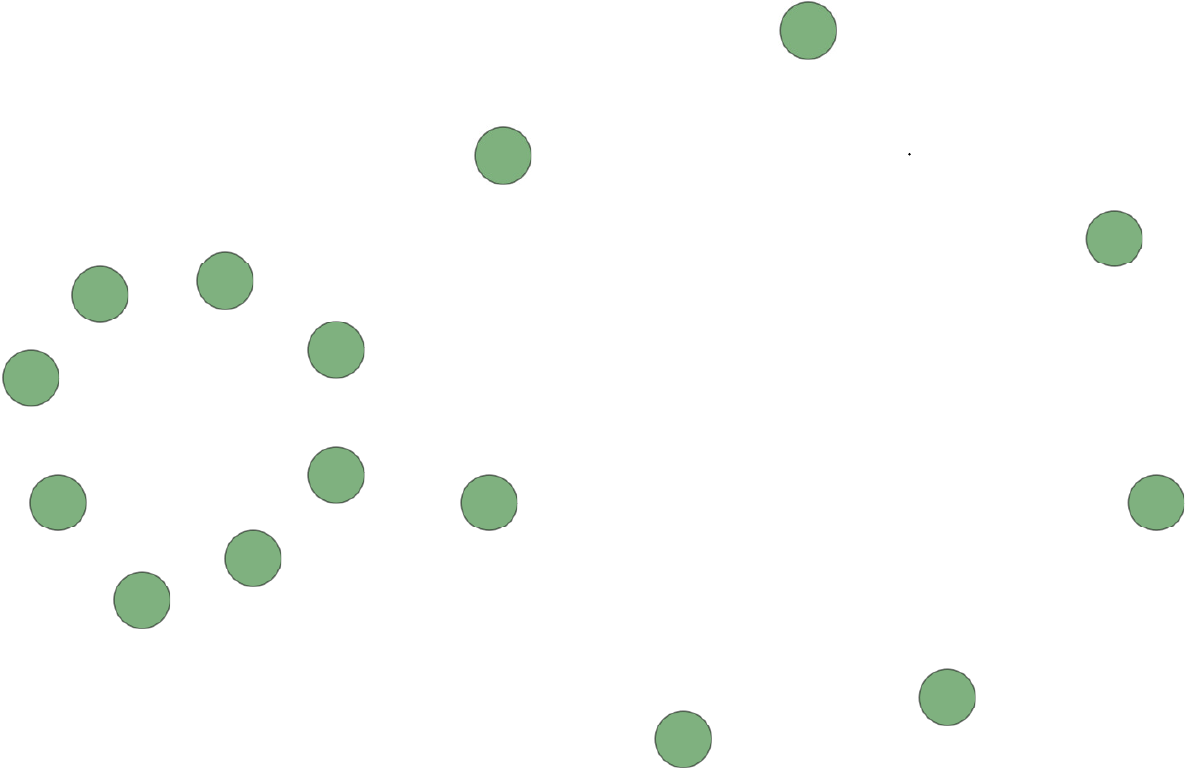}
\end{subfigure}
\begin{subfigure}[t]{0.3\columnwidth}
\centering
\includegraphics[width=0.65\textwidth]{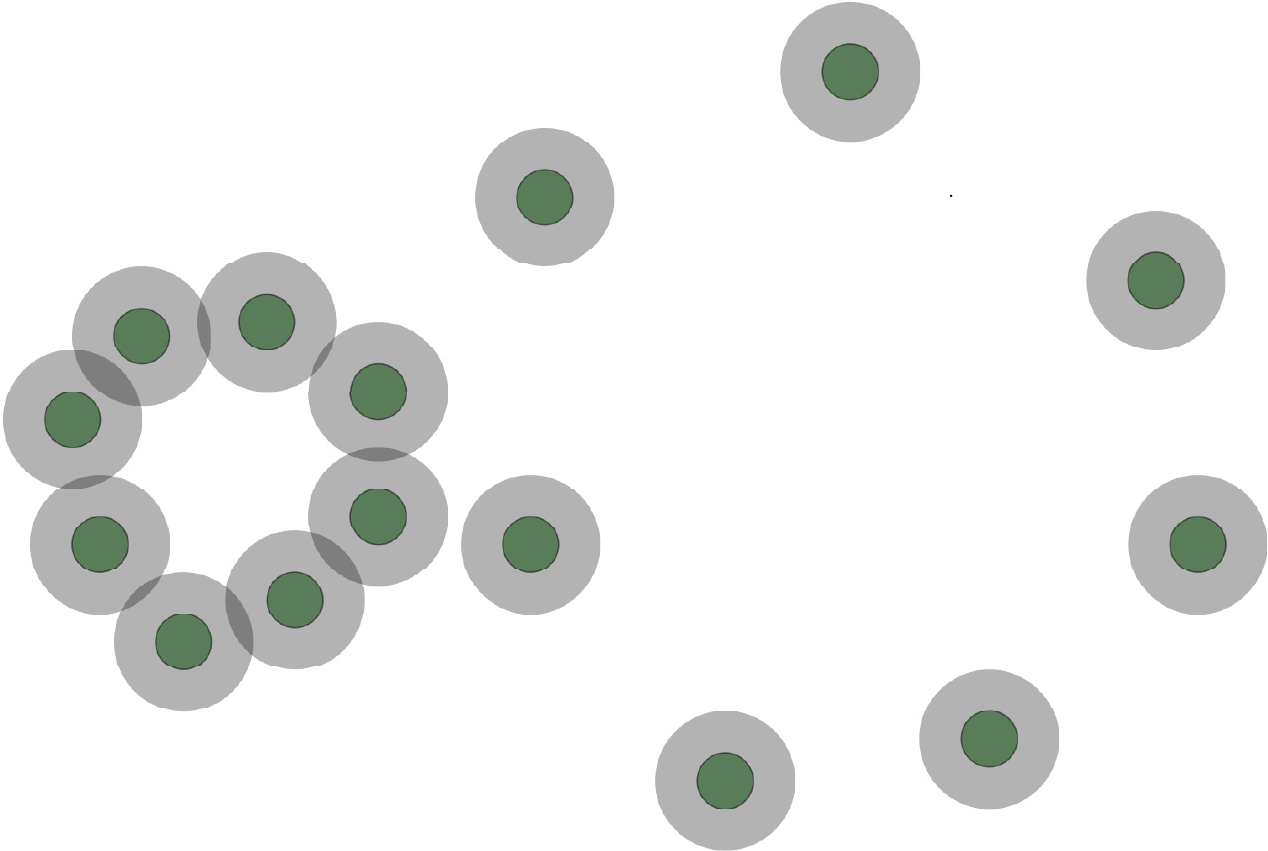}
\end{subfigure}
\begin{subfigure}[t]{0.3\columnwidth}
\centering
\includegraphics[width=0.65\textwidth]{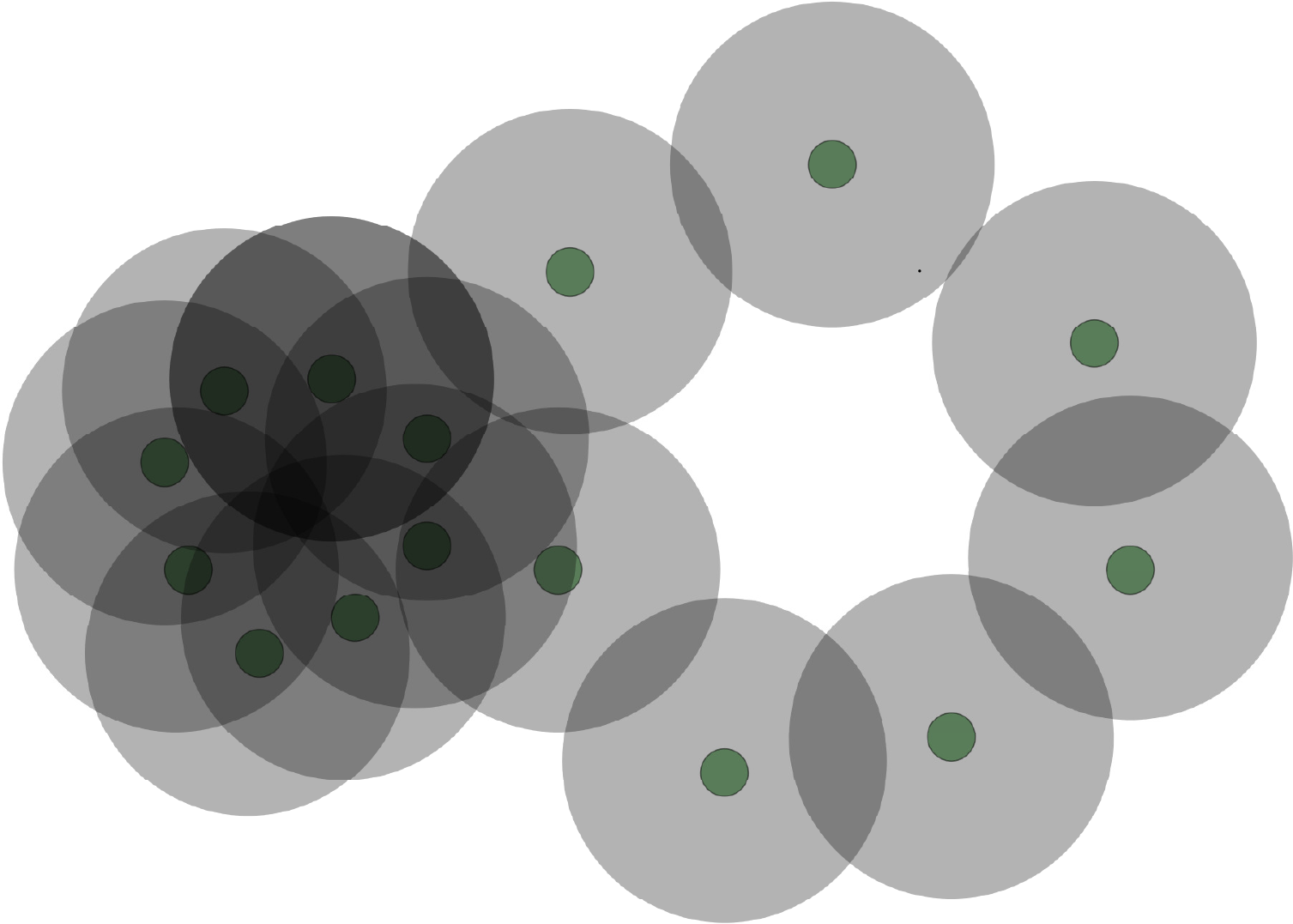}
\end{subfigure}
\caption{At no scale, can the union of balls 
determine the two loops simultaneously.
}
\label{figure:union_of_balls}
\end{figure}

One answer, formalized by the theory of \emph{persistent homology}, is
to generate a summary of every topological feature formed at any scale. Note that $\B_\alpha\subseteq\B_{\alpha'}$ for $\alpha\leq\alpha'$.
The collection $(\B_\alpha)_{\alpha\geq 0}$ together with 
inclusion maps $\B_\alpha\to\B_{\alpha'}$ whenever
$\alpha\leq\alpha'$ forms a \emph{filtration}. The topological properties
of this filtration can be summarized by a finite multi-set of points in the
plane, called the \emph{persistence diagram}. For every topological feature, 
a point in the diagram denotes the range of scales for which the
feature is present. This range is called the \emph{persistence}
of the feature.

Importantly, there are several notions of a distance between two diagrams.
In this paper, we use the \emph{multiplicative bottleneck distance} 
to compare persistence diagrams.
Intuitively, diagrams with small distance have similar features that have 
long persistence; note that this does not imply that their homology is equal
on any scale (see Section~\ref{section:preliminaries} for details).

Another important aspect of the above theory is the possibility of efficiently
computing persistence diagrams. 
The \emph{\Cech complex} $\cech_\alpha$ is the intersection/nerve complex
of the balls of radius $\alpha$ and has the same homological properties
as the union of balls.
The derived \emph{\Cech filtration} is a combinatorial object
yielding the same persistence diagram as the union of balls. 
The problem with
this standard approach is the sheer size of the produced
filtration. It is common practice to only look at the $k$-skeleton
of the \Cech filtration, that is, ignore intersection of $(k+2)$ or
more balls, where $k\leq d$ is another constant parameter. 
That filtration consists of $O(n^{k+1})$ simplices, which is too large
for large data sets, even for relatively small values of $k$.
In this paper, we will present a compact and efficiently computable 
approximation of the \Cech filtration. 

\paragraph{Problem statement and prior work.}
We study the question: 
can we  efficiently compute a discrete filtration
$(A_{\alpha_i})_{i=1,\ldots,\ell}$, called the \emph{approximate filtration},
that is significantly smaller than the \Cech filtration with the 
persistence diagrams of \Cech and approximate filtration being close to each
other? Note that the question defines three quality criteria: the size, computation time
and the quality of the approximation.

Sheehy~\cite{sheehy-rips} showed that an $\eps$-approximate filtration
of size $n(\frac{1}{\eps})^{O(dk)}$ can be computed efficiently for the $k$-skeleton. For $k=d$, the size of this filtration is $n(1/\eps)^{O(d^2)}$.
His result, formulated for the closely related \emph{Vietoris-Rips complex}, also extends to the \Cech complex.
Note that the size is linear in the number of points, which 
is a significant  improvement over $O(n^{k+1})$ especially when $n$ is large and $d$, $k$,
and $1/\eps$ are of moderate size. The technical idea of the approximation
scheme is to compute a net-tree based hierarchical clustering and using it to combine balls whose centers are in the  same cluster into a single ball. 
As a result of this, a large number of ball intersections can be avoided.
Similar ideas with the same guarantees, have appeared
in~\cite{bs-approximating,dfw-gic,ks-wssd}.
The dimension $d$ in the complexity bounds can also be replaced with the
doubling dimension of the metric space, which makes 
the result especially useful for spaces with low doubling dimension.

Choudhary et al.~\cite{ckr-polynomial-dcg} presented an alternative
approximation scheme, where the approximation quality is only $O(d)$,
but the size of the approximation is 
significantly smaller at  $n2^{O(d\log k)}$. 
This result is obtained by tiling the space into \emph{permutahedra}
(a generalization of the hexagonal grid
in $\R^2$) where the diameter of the permutahedra is controlled
by the scale parameter $\alpha$. Every permutahedron containing at least one input point is selected and the nerve of every selected permutahedron is reported as the approximation. The improvement in size
stems from the fact that at most $(d+1)$-permutahedra can intersect
in $d$ dimensions which upper bounds the number of simplices in the nerve.
The approximation quality has been improved to $O(\sqrt[4]{d})$
in subsequent work, with size  $n2^{O(d\log k)}$~\cite{ckr-barycentric}.
In~\cite{ckr-polynomial-dcg}, it is also shown that, for any $\eps < 1/\log^{1+c} n$ with $c\in (0,1)$, any $\eps$-approximate filtration must have a size of at least $n^{\Omega(\log \log n)}$.
\paragraph{Our contribution.}
We derive an approximation scheme of size 
\[n\, 2^{O(d\log d+dk)} \left(\frac{1}{\eps}\right)^{O(d)}\]
for an $\eps$-approximation of the \Cech complex for $\eps\in(0,1]$.
For constant dimension $d$, this simplifies to $n\left(\frac{1}{\eps}\right)^{O(d)}$,
improving all previous results for the Euclidean case. 
We achieve this approximation based on the techniques devised in
\cite{ckr-barycentric,ckr-polynomial-dcg}.
On a fixed scale $\alpha$, we tile the space with a cubical grid,
carefully select a subset of them and take their nerve as our approximation complex.
The novelty lies in the resolution of the tiling. 
Unlike in~\cite{ckr-polynomial-dcg} where the diameter of a permutahedron
was in the order of $\alpha$, we use a much smaller diameter of roughly
$\eps\alpha$, and we select a hypercube in the approximation
if its center is $\alpha$-close to an input point. This is equivalent to
approximating the union of $\alpha$-balls $\B_\alpha$
with cubical ``pixels'' at resolution $\eps$ (see Figure~\ref{figure:digital_2d}).
The number of selected pixels can be bounded by  
$n\left(\frac{d}{\eps}\right)^d$ which is significantly larger 
than the number of points. 
Perhaps counter-intuitively, its nerve is still smaller
than in previous approaches, because each vertex of the nerve is only
incident to $2^{O(dk)}$ simplices in the $k$-skeleton.
This technique resembles previous work in computational geometry
where adding points (referred to Steiner points) helps to 
reduce the size of triangulations~\cite{beg-provably}.
Hudson et al.~\cite{hmos-meshing} used similar techniques
to compute an $\eps$-approximation 
of the \Cech complex of size $n/\eps^{O(d^2)}$.

Our results pave the way for important extensions whose proofs 
we omit for the sake of brevity of the presentation. 
We only announce them to underline the importance of our techniques
and postpone the technical discussion to 
an extended version of this manuscript:

\begin{itemize}
\item Our approximation from above is not a filtration, but a simplicial tower,
that is a sequence of simplicial complexes connected by simplicial maps.
Our technique allows, however, for a definition of an actual approximate
filtration of the same size. This filtration has the additional property
of being a flag complex at each scale, which has computational benefits
as its persistence diagram can be computed using faster algorithms.%
\footnote{e.g., Ripser by U.Bauer: \url{https://github.com/Ripser/ripser}}

\item The factor $2^{O(d\log d+dk)}$ in our size bound can be further reduced
to $2^{O(d\log d)}$ by replacing the cubical grid by a permutahedral grid~\cite{ckr-polynomial-dcg}.
A consequence of this result is that when $d = \Theta(\log n)$ and $\eps = 1/\log^{1+c} n$, for some $c \in (0,1)$, 
our approximation with the permutahedral grid has a size of $n^{O(\log \log n)}$ matching the lower bound of~\cite{ckr-polynomial-dcg}.
\end{itemize}  

The result of this work are proved in two steps. We first devise a simple
approximation scheme which approximates the union of balls by pixels
of a carefully chosen size on each scale. It follows with standard techniques
that the approximation quality is $(1+\eps)$. However, the total size of
the approximation has an additional factor that depends on the spread
of the point set. To achieve spread-independence, we use the observation
that when we increase the scale from $\alpha$ to $(1+\eps)\alpha$, the growth
in radius of a ball does only affect the topology of the union of balls
if the enlarged balls gives rise to a new intersection pattern with other 
balls.
In other words, we can delay the growth of a ball (and hence, avoiding
resampling the balls by pixels) if $\alpha$ is not in a critical range
of scales. Moreover, we can bound the critical range of all balls
efficiently through a \emph{well-separated pair decomposition} (WSPD)
of the point set. This suffices to reduce the total number of pixels used
in the approximation to $n(\frac{1}{\eps})^{O(d)}$.
One technical difficulty arising from this approach is that on a fixed scale,
the approximation complex consists of cubes of different sizes.

\paragraph{Outline of the paper}
We review the topological concepts needed for our result 
in Section~\ref{section:preliminaries}.
We present the simple, spread-dependent approximation in
Section~\ref{section:digitization}, and extend it to a spread-independent
approximation in Section~\ref{section:spreadremoval}. We conclude
in Section~\ref{section:conclusion}.

\section{Preliminaries}
\label{section:preliminaries}

We give a short introduction to topological concepts that are essential
for our results.
For more details, we refer to standard references such as~\cite{bss-metrics,cdgo-sspm,eh-book,hatcher,munkres}.

\subparagraph*{Simplicial complexes and simplicial maps}

A \emph{simplicial complex} $K$ on a finite set of elements $S$ 
is a collection of subsets $\{\sigma\subseteq S\}$ called \emph{simplices} 
such that each subset $\tau\subset\sigma$ is also in $K$. 
In other words, simplicial complexes are closed under taking subsets.
We say that a simplex $\sigma\in K$ has dimension $k:=|\sigma|-1$,
in which case $\sigma$ is called a \emph{$k$-simplex}.
A simplex $\tau$ is a \emph{face} of $\sigma$ if $\tau\subseteq\sigma$. 
The \emph{$k$-skeleton} of $K$ consists of
all simplices of $K$ of dimension $k$ or lower.
For instance, the $1$-skeleton of $K$ is a graph 
defined by its $0$-simplices and $1$-simplices.

For a point set $P\subset\R^d$ and a real number $\alpha\ge 0$,
the \emph{\Cech complex} $\cech_\alpha$ on $P$ at scale $\alpha$ is the collection of all 
simplices $\sigma\subseteq P$ such that there is a non-empty
intersection between the Euclidean balls of radius $\alpha$ centered at
the points of $\sigma$. 
Equivalently, $\sigma\in\cech_\alpha$ if the minimum enclosing ball
of $\sigma$ has radius at most $\alpha$.
We denote the union of $\alpha$-balls as $\B_\alpha$.
See Figure~\ref{figure:only_cech} for an example.

\begin{figure}[h]
\centering
\includegraphics[width=0.3\columnwidth]{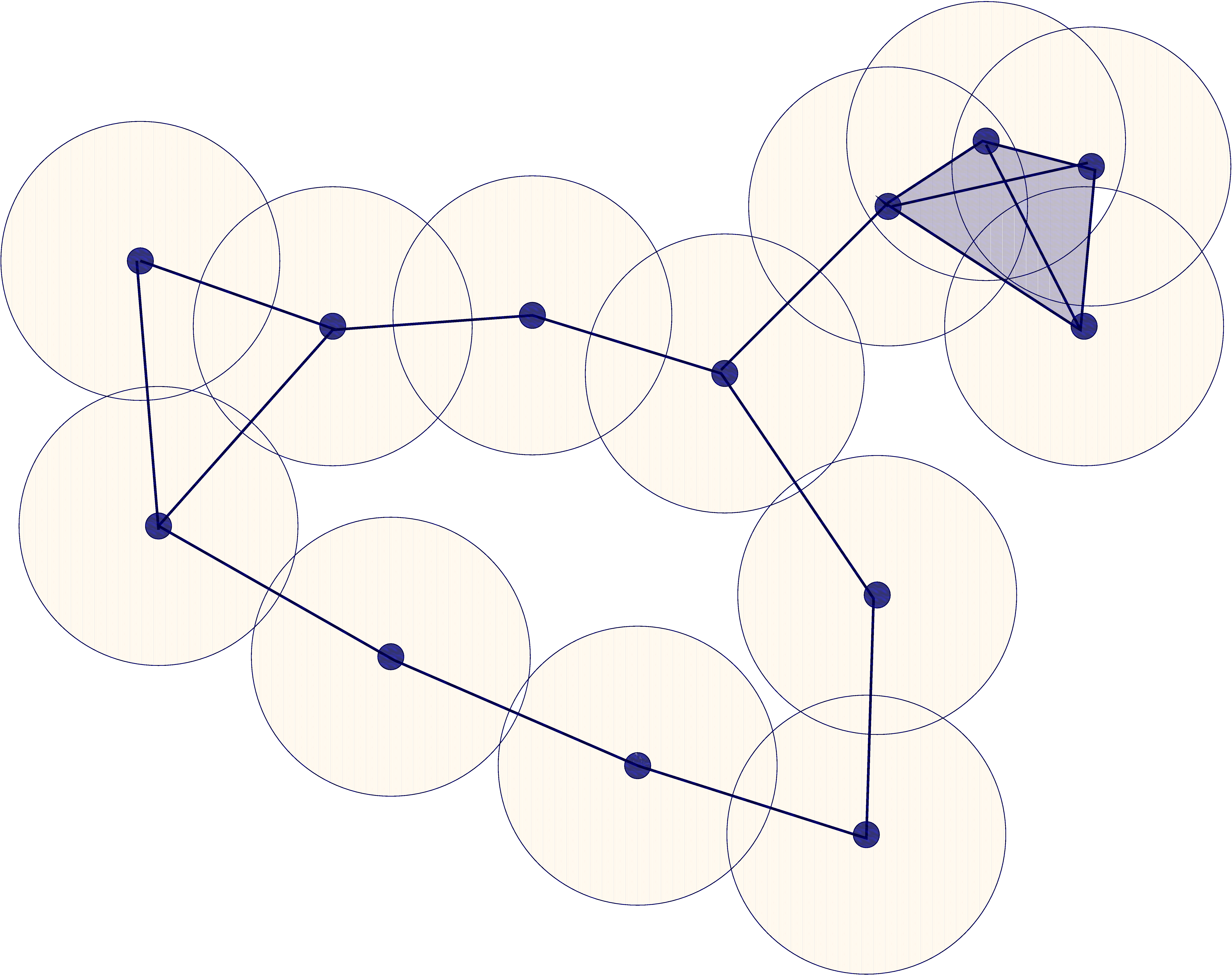}
\caption{
$\B_\alpha$ and $\cech_\alpha$ for some $\alpha>0$.
}
\label{figure:only_cech}
\end{figure}

\begin{figure}[t!]
\centering
\begin{subfigure}[t]{0.4\columnwidth}
\centering
\includegraphics[width=0.6\textwidth]{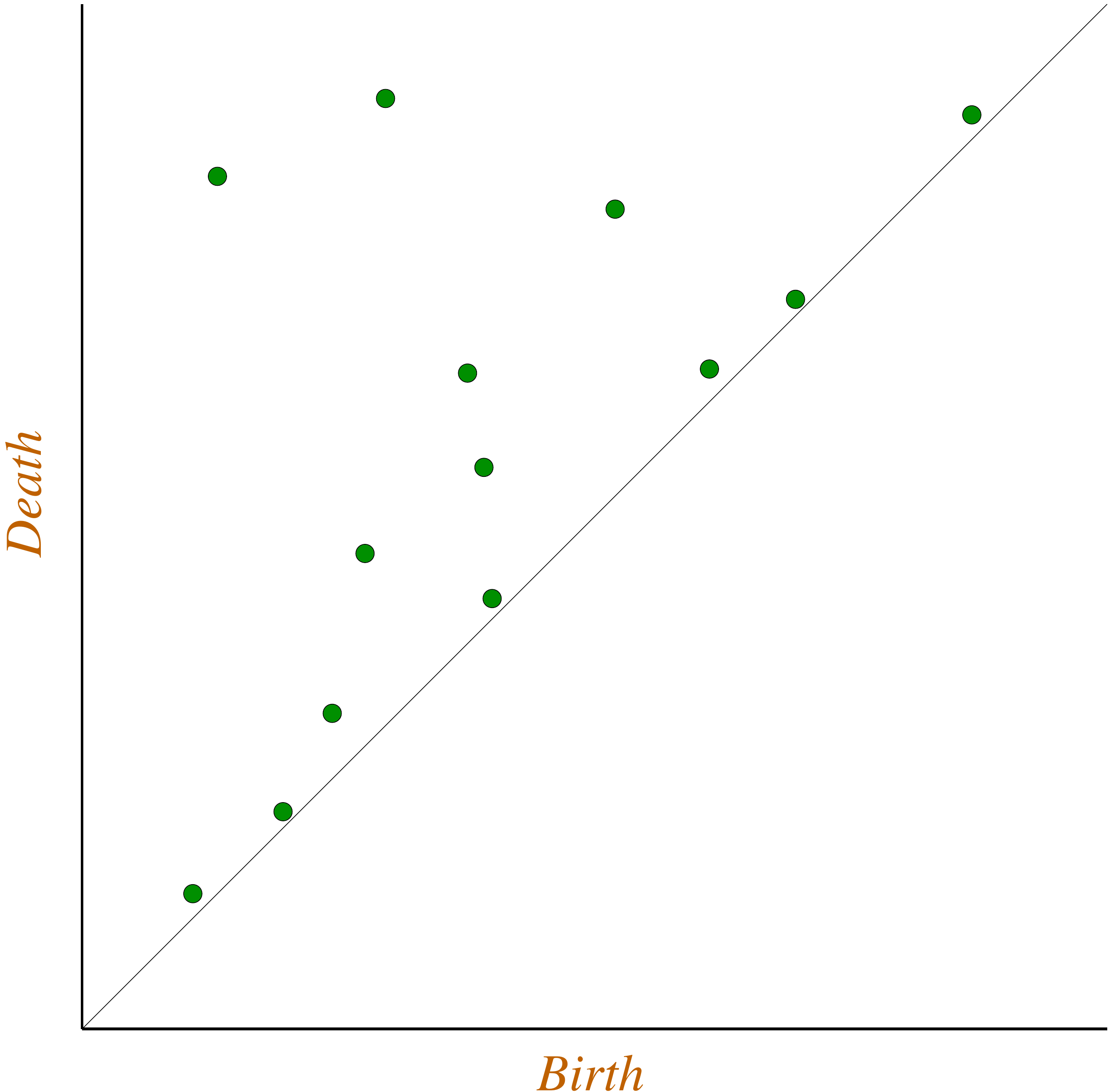}
\end{subfigure}
\begin{subfigure}[t]{0.4\columnwidth}
\centering
\includegraphics[width=0.6\textwidth]{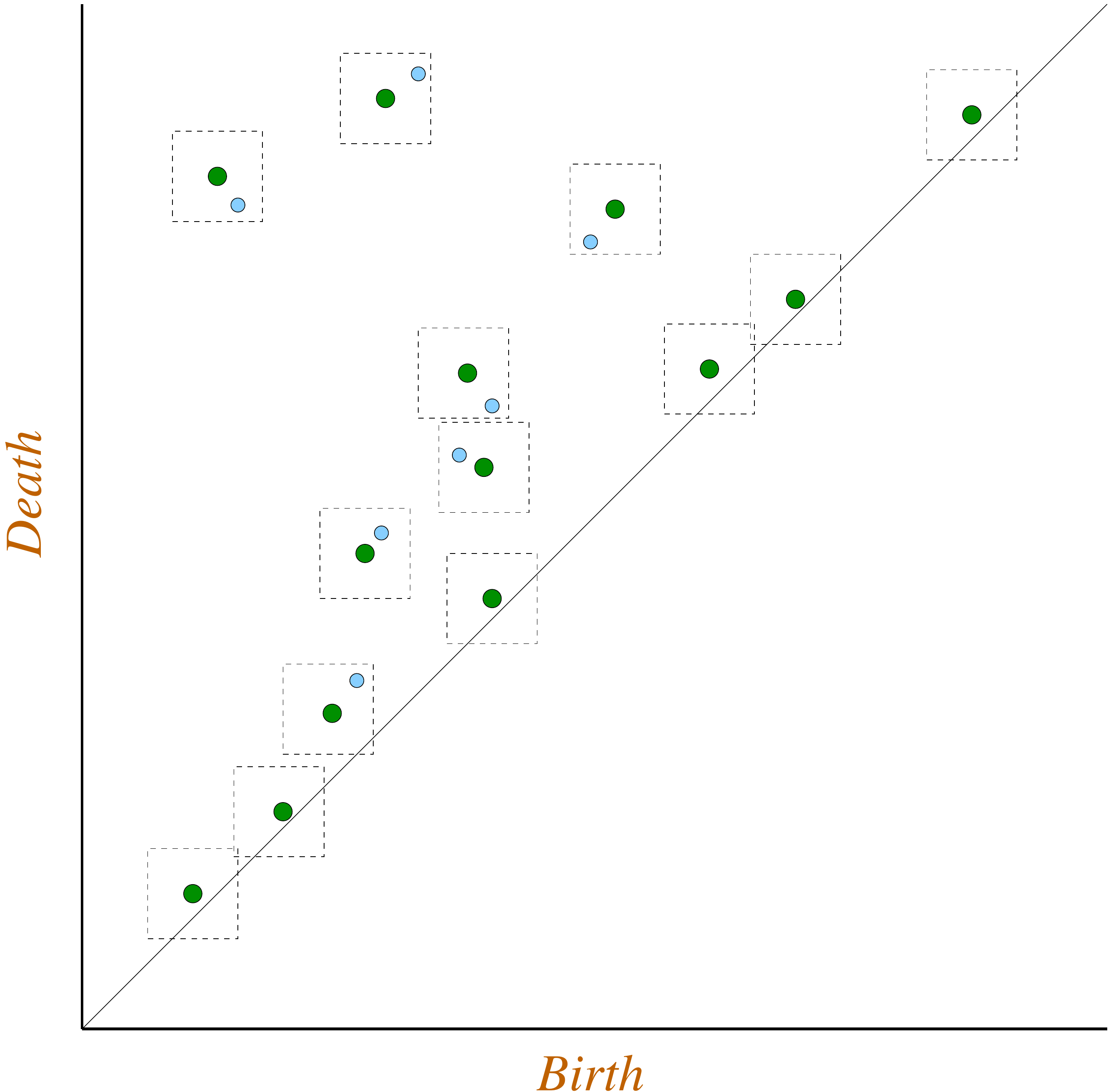}
\end{subfigure}
\caption{A persistence diagram on the left, and its approximation 
on the right.
Points within boxes form a partial matching.}
\label{figure:pd}
\end{figure}

A simplicial complex is called a \emph{flag complex}, 
if it has the following property:
whenever a set $\{p_0,\ldots,p_k\}\subseteq P$ has the property
that every $1$-simplex $\{p_i,p_j\}$ is in the complex, then the
$k$-simplex $\{p_0,\ldots,p_k\}$ is also in the complex.

A simplicial complex $K'$ is a \emph{subcomplex} of $K$ if $K'\subseteq K$.
For instance, $\cech_{\alpha}$ is a subcomplex of $\cech_{\alpha'}$ for 
$0\le\alpha\leq\alpha'$. 
Let $L$ be another simplicial complex.
Let $\hat{\simplicialmap}$ be any map that assigns to each vertex of $K$
a vertex of $L$.
A \emph{simplicial map} is a map $\simplicialmap:K\rightarrow L$ defined
using some vertex map $\hat{\simplicialmap}$, such that
for every simplex $\sigma=\{p_0,\ldots,p_k\}$ in $K$, the set of vertices
$\simplicialmap(\sigma):=\{\hat{\simplicialmap}(p_0),\ldots,\hat{\simplicialmap}(p_k)\}$ 
is a simplex of $L$. 
For $K'\subseteq K$, the inclusion map $\inc:K'\rightarrow K$
is an elementary example of a simplicial map. 
Note that any simplicial map is completely specified
by its action on the vertices of the domain.

\paragraph{Homotopy and Nerve theorem}
Let $f_1,f_2:X\rightarrow Y$ be two continuous maps between 
topological spaces $X,Y$.
A continuous function 
$H:X\times[0,1]\rightarrow Y$ is said to be a \emph{homotopy} 
between $f_1$ and $f_2$, 
if $H(x,0)=f_1(x)$ and $H(x,1)=f_2(x)$. 
In this case, $f_1$ and $f_2$ are said to be \emph{homotopic} to each other, 
and we record this relation as $f_1\homt f_2$.
Informally, the second parameter of $H$ can be interpreted as time,
so that $H$ is a continuous deformation of $f_1$ into $f_2$, 
as time varies from 0 to 1.

$X$ and $Y$ are said to be \emph{homotopy equivalent} if there exist 
continuous maps $f:X\rightarrow Y$ and $g:Y\rightarrow X$ such that 
$f\circ g\homt \id_Y$ and $g\circ f\homt \id_X$, where $\id_X$ and $\id_Y$ 
are the identity functions on $X$ and $Y$, respectively.
We record this relation as $X\homt Y$.
Intuitively, two spaces are homotopy equivalent if they can be continuously
transformed into one another.

Let $U:=\{U_1,\ldots,U_m\}$ denote a finite collection of sets.
The \emph{nerve} of $U$ is an abstract simplicial complex $\nrv(U)$ 
consisting of simplices
corresponding to non-empty intersections of elements of $U$, that is,
\[
\nrv(U):=\{V\subseteq U \mid  \bigcap_{U_i\in V} U_i \neq \varnothing\}.
\]
For instance, the \Cech complex at scale $\alpha$ 
is $\cech_\alpha=\nerve(\B_\alpha)$.

\begin{theorem}[nerve theorem]
\label{theorem:nervethm}
Let us denote by $U:=\{U_1,\ldots,U_m\}$ a finite collection of closed sets 
in Euclidean space, such that all non-empty intersections of the form
$\{\cap_{U_i\in V} U_i \neq \varnothing \mid V\subseteq U  \}$
are contractible.
Then, $\nrv(U)\homt \bigcup_{i=1}^{m} U_i$, that is, they are homotopy 
equivalent (see \cite{bjorner-detailed,bjorner-small,borsuk-nerve,walker-hom} 
for more general versions of the theorem).
\end{theorem}

In particular, if $U$ is a collection of convex sets, then all non-empty
common intersections are contractible, so that the nerve theorem
applies to $U$.
As an example, $\cech_\alpha\homt\B_\alpha$.

\subparagraph*{Filtrations and simplicial towers}
Let $I\subseteq\R$ be a set of real values which we refer to as \emph{scales}.
A \emph{filtration} is a collection of simplicial complexes indexed by $I$,
$(K_\alpha)_{\alpha\in I}$ such that $K_\alpha\subseteq K_\alpha'$ 
for all $\alpha\leq\alpha'\in I$. 
For instance, $(\cech_\alpha)_{\alpha\geq 0}$ is the \emph{\Cech filtration}.
A \emph{(simplicial) tower} is a sequence $(K_\alpha)_{\alpha\in J}$ of simplicial 
complexes with $J$ being a discrete set (for instance $J=\{2^k\mid k\in\Z\}$),
together with simplicial maps 
$\simplicialmap_\alpha:K_\alpha\rightarrow K_{\alpha'}$ between
complexes at consecutive scales.

A tower of the form $(K_\alpha)_{\alpha\in J}$
on simplicial maps $g^{.,.}$ with index set $J=\{0\le j_1<j_2<\ldots\}$
can always extend the tower to $(K_\alpha)_{\alpha\ge 0}$,
by setting $K_\alpha=K_{j_i}$ 
and $g^{\alpha,j_{i+1}}=g^{j_i,j_{i+1}}$ 
for all $\alpha\in[j_i,j_{i+1})$.
We call this technique of extending towers the 
\emph{standard filling technique} in our paper for simplicity.
The approximation constructed in this paper will be an 
example of such a tower.

We say that a simplex $\sigma$ is \emph{included} in the tower at scale 
$\alpha'$ if $\sigma\in K_{\alpha'}$ is not in the image of the map
$\simplicialmap_{\alpha}:K_\alpha\rightarrow K_{\alpha'}$,
where $\alpha$ is the scale preceding $\alpha'$ in the tower.
The \emph{size} of a tower is the number of simplices included over all scales.
If a tower arises from a filtration, its size is simply the size of the 
largest complex in the filtration (or infinite, if no such complex exists).
However, this is not true in general for
simplicial towers, since simplices can collapse in the tower and the size of the
complex at a given scale may not take into account the collapsed simplices
which were included at earlier scales in the tower.

\subparagraph*{Persistence diagram and Interleavings}
A collection of vector spaces $(V_\alpha)_{\alpha\in I}$ connected with homomorphisms (linear maps)
$\linearmap_{\alpha_1,\alpha_2}:V_{\alpha_1}\rightarrow V_{\alpha_2}$ 
is called a \emph{persistence module}, if $\linearmap_{\alpha,\alpha}$ is 
identity on $V_\alpha$ for all $\alpha\in I$, and 
$\linearmap_{\alpha_2,\alpha_3}\circ\linearmap_{\alpha_1,\alpha_2}
=\linearmap_{\alpha_1,\alpha_3}$ 
for all $\alpha_1\le\alpha_2\le\alpha_3\in I$ for the index set $I$.

We can generate persistence modules using simplicial complexes.
Given a simplicial tower $(K_\alpha)_{\alpha\in I}$, we can fix some base
field $\field$ to obtain a sequence $(H(K_\alpha))_{\alpha\in I}$
of vector spaces with linear maps $\overline{\simplicialmap}^\ast$, 
that forms a persistence module.
This is true because of the functorial properties of homology~\cite{munkres}.
The same construction is also applicable to filtrations.

Under certain tameness conditions,
persistence modules admit a decomposition into a collection of 
intervals of the form $[\alpha,\beta]$
(with $\alpha,\beta\in I$).
These intervals can be represented as a set of points in the plane, 
called the \emph{persistence diagram}, where each interval $[\alpha,\beta]$
is simply represented as the point $(\alpha,\beta)$.
See Figure~\ref{figure:pd} for an example.
The persistence diagram of a persistence module characterizes the 
module uniquely up to isomorphism.
If the persistence module is generated by a simplicial complex,
each point $(\alpha,\beta)$ in the diagram corresponds 
to a homological feature (a ``hole'') that comes into existence at 
complex $K_\alpha$ and persists until it disappears at $K_\beta$. 

Two persistence modules $(V_\alpha)_{\alpha\ge 0}$ and $(W_\alpha)_{\alpha\ge 0}$
with respective linear maps $\phi_{\cdot,\cdot}$ and $\psi_{\cdot,\cdot}$
are said to be \emph{(multiplicatively) strongly $c$-interleaved} 
if there exist a pair of families of 
linear maps $\gamma_\alpha:V_{\alpha}\rightarrow W_{c\alpha}$ and 
$\delta_\alpha:W_{\alpha}\rightarrow V_{c\alpha}$ for $c>0$,
such that Diagram~\eqref{diag:strong_diag} 
commutes for all $0\le \alpha\le \alpha'$ (the subscripts of the maps
are excluded for readability).
In such a case, the persistence diagrams of the two modules are said 
to be $c$-approximations of each other in the sense of~\cite{ccggo-proximity}.
More precisely, there is a partial matching between the points of the 
two diagrams, after a suitable logarithmic scaling of the diagrams. 
Further details can be found in~\cite{buchet-thesis}.
See Figure~\ref{figure:pd} for an example.

\begin{equation}
\label{diag:strong_diag}
\xymatrix{
	V_{\frac{\alpha}{c}} \ar[rrr]^{\phi} \ar[rd]^{\gamma} & & & V_{c\alpha'}    &  & V_{c\alpha} \ar[r]^{\phi} & V_{c\alpha'} 
	\\
	& W_\alpha \ar[r]^{\psi} & W_{\alpha'} \ar[ru]^{\delta} &                 & W_\alpha \ar[r]^{\psi} \ar[ru]^{\delta} & W_{\alpha'} \ar[ru]^{\delta}
	\\ 
	& V_\alpha \ar[r]^{\phi} & V_{\alpha'} \ar[rd]^{\gamma} &                 & V_\alpha \ar[r]^{\phi} \ar[rd]^{\gamma} & V_{\alpha'} \ar[rd]^{\gamma} 
	\\ 
	W_{\frac{\alpha}{c}} \ar[rrr]^{\psi} \ar[ru]^{\delta} & & & 
	W_{c\alpha'}    &  & W_{c\alpha} \ar[r]^{\psi} & W_{c\alpha'}\\
}
\end{equation}

Next, we discuss a special case
that relates to the equivalence of persistence modules~\cite{cz-computing,handbook}.
Two persistence modules 
$\V=(V_\alpha)_{\alpha\in I}$ and $\W=(W_\alpha)_{\alpha\in I}$
on linear maps $\phi,\psi$ respectively are isomorphic,
if there exists an isomorphism $f_\alpha:V_\alpha\rightarrow W_\alpha$ for
each $\alpha\in I$ such that the diagram
\begin{eqnarray}
\label{equation:persisistence_equiv}
\xymatrix{
	\dots \ar[r] & V_\alpha \ar[r]^{\phi} \ar[d]^{f_\alpha} 
	& V_\beta \ar[r]  \ar[d]^{f_\beta} & \dots 
	\\
	\dots \ar[r] & W_\alpha \ar[r]^{\psi} & W_\beta \ar[r] & \dots 
}
\end{eqnarray}
commutes for all $\alpha\le \beta \in I$.
Isomorphic persistence modules have identical persistence barcodes.

\section{A simple digitization scheme}
\label{section:digitization}

In this section we describe our first approximation scheme which is
conceptually simple and lays down the foundation for the more technical
approximation scheme that we present in Section~\ref{section:spreadremoval}.
We describe our scheme for a set of $n$ points $P$ in 
$d$-dimensional Euclidean space. 
We assume throughout that $\eps\in(0,\frac{1}{5}]$.

\paragraph{A cubical subdivision}

Let $J:=\{2^{i} \mid i\in\Z \}$.
For any $\alpha\in J$, we define the lattice 
$L_\alpha:=\left(\frac{\eps\alpha}{4\sqrt{d}}\right)\Z^d$ 
as the $\Z^d$ lattice whose basis vectors have been scaled by
$\frac{\eps\alpha}{4\sqrt{d}}$.
We define the \emph{pixels} of this lattice as the smallest $d$-dimensional cubes 
whose vertices are the lattice points.
The diameter of the pixels of $L_\alpha$ is at most $\eps\alpha/4$. 
Our approximation complexes are built as nerves of the union
of pixels of $L_\alpha$.\footnote{It seems simpler at first to define
the approximation as a cubical complex directly instead of taking the nerve, 
but it is more 
complicated to construct the chain maps connecting different scales
in this cubical setup.}

Each pixel of lattice $L_\alpha$ is fully contained
in some pixel of $L_\beta$ when $\beta>\alpha$ with $\beta,\alpha$ both in $J$.
This also implies that each pixel center of $L_\alpha$ has a unique
nearest-neighbor among the pixel centers in $L_\beta$.

\subsection{Approximation complex}
\label{subsection:appcpx}

Let $B(p,\alpha)$ denote the $d$-dimensional Euclidean ball of radius $\alpha$,
centered at any input point $p\in P$.
We denote by
\[
\B_\alpha:= \left( B(p,\alpha) \right)_{p\in P},
\]
the set of $\alpha$-balls centered at the input points. 
Naturally, the \Cech complex on $P$ at scale $\alpha$ is 
the nerve of $\B_\alpha$, that is, 
$\cech_\alpha=\nrv(\B_\alpha)$.

Let $I$ denote the set of scales
\[
I:=\{ \alpha_k:=(1+\eps)^k \mid k\in \Z \}.
\]
We now define our approximation complex $\complex_\alpha$ for $\alpha\in I$.
Consider the maximal $\beta\in J$ such that $\beta\leq\alpha$.
As discussed, there exists a cubical subdivision based on the lattice
$L_\beta=(\frac{\eps\beta}{4\sqrt{d}})\Z^d$.
Let $\pixels_\alpha$ denote the set of pixels of that subdivision whose
center lies in $\B_\alpha$, that is, which are in distance at most $\alpha$
to a point in $P$. We define $\complex_\alpha:=\nrv(\pixels_\alpha)$.
Also, we write $|\pixels_\alpha|$ for the union of all pixels in $\pixels_\alpha$.
See Figure~\ref{figure:digital_2d} for an illustration
of $\B_\alpha$ and $|\pixels_\alpha|$.

\begin{figure}[t!]
\centering
\begin{subfigure}[t]{0.5\columnwidth}
\centering
\includegraphics[width=0.6\textwidth]{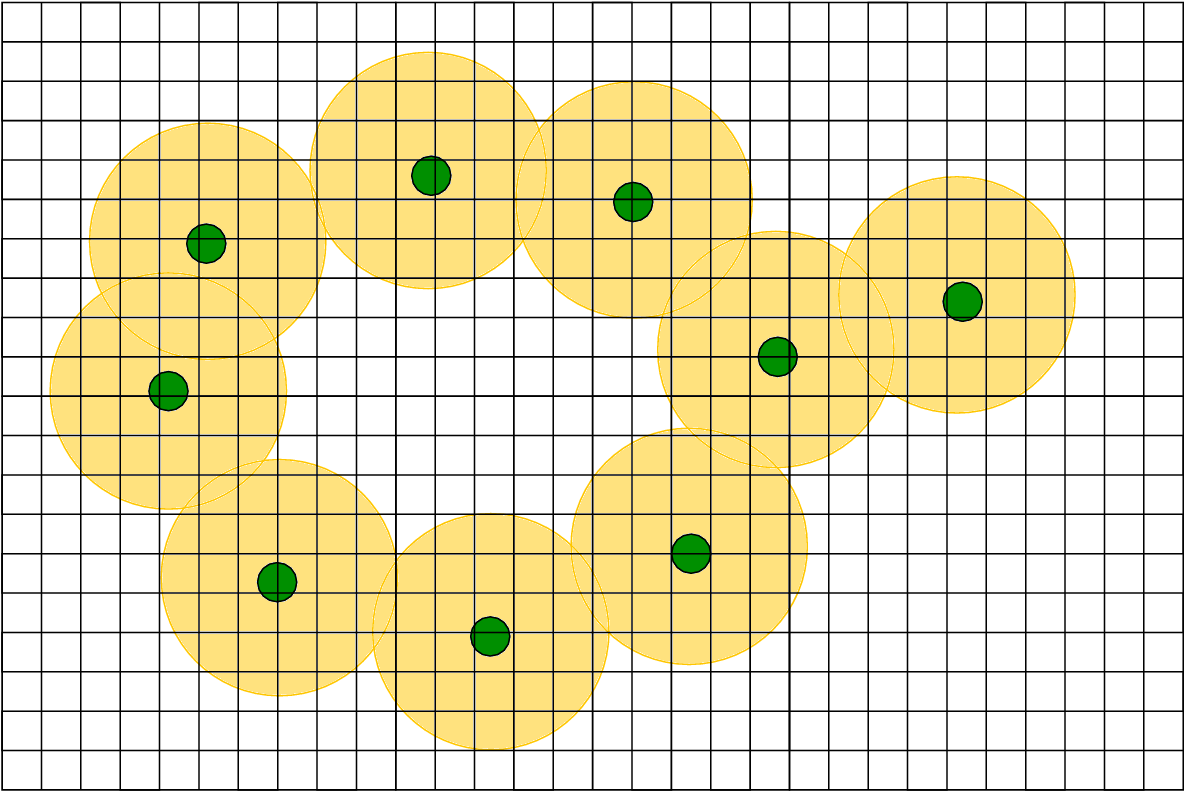}
\caption{
A collection of balls centered
at the input points...
}
\end{subfigure}%
\begin{subfigure}[t]{0.5\columnwidth}
\centering
\includegraphics[width=0.6\textwidth]{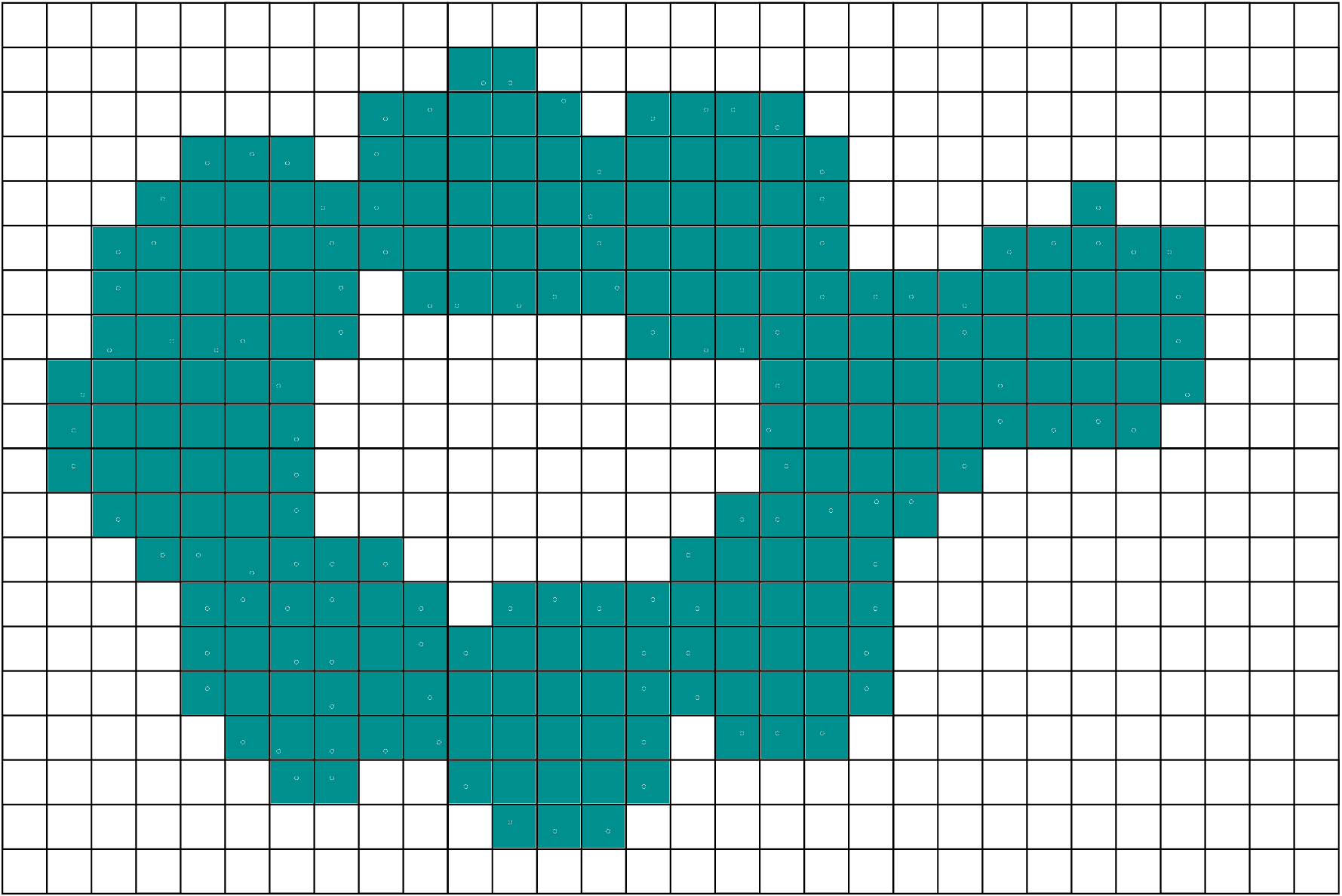}
\caption{
...and the corresponding pixels.
}
\end{subfigure}
\caption[Digitization]{An example of digitization in the plane.}
\label{figure:digital_2d}
\end{figure}

\begin{lemma}[Sandwich lemma]
\label{lemma:digital_sandwich}
For any $\alpha\in I$, 
\[
|\pixels_\alpha|\subseteq \B_{(1+\eps/2)\alpha}\subseteq |\pixels_{(1+\eps)\alpha}|.
\]
\end{lemma}

\begin{proof}
For $\alpha\in I$, let $\beta\in J$ be as in the definition of 
$|\pixels_\alpha|$. For the first inclusion, if $x\in|\pixels_\alpha|$,
$x$ lies in some pixel $C$ of the grid $L_\beta$.
Let $c$ denote the center of this pixel.
There is some $p\in P$ with $\distance{c}{p}\leq\alpha$.
The diameter of $C$ is $\eps\beta/4<\eps\alpha$, hence 
$\distance{c}{x}\leq \eps\alpha/2$, and by the triangle inequality,
$\distance{x}{p}\leq (1+\eps/2)\alpha$.

For the second inclusion, fix $x\in \B_{(1+\eps/2)\alpha}$
and let $p\in P$ such that $\distance{x}{p}\leq (1+\eps/2)\alpha$.
Note that $(1+\eps)\alpha\in I$ as well,
and the approximation complex is either constructed using the same $\beta$
as for $\alpha$, or using $2\beta$ if $2\beta=(1+\eps)\alpha$.
In both cases, $x$ lies in a pixel $C$ with diameter at most 
$\eps\beta/2<\eps\alpha$, so the distance of $x$ to the center $c$ of $C$
is at most $\eps\alpha/2$. 
By the triangle inequality, $\distance{x}{p}\leq (1+\eps)\alpha$, 
so the pixel belongs to $|\pixels_{(1+\eps)\alpha}|$.
\end{proof}

The lemma shows that $|\pixels_\alpha|\subset|\pixels_{(1+\eps)\alpha}|$.
Using the standard filling technique from Section~\ref{section:preliminaries},
we extend this discrete space to a continuous filtration:

\[(|\pixels_{\alpha}|)_{\alpha > 0}.\]

Moreover, the Sandwich lemma together with the 
strong-interleaving diagram (Diagram~\eqref{diag:strong_diag})
implies at once:

\begin{theorem}
\label{theorem:digital_basic_intlv}
$(H(|\pixels_{\alpha}|))_{\alpha> 0}$
$(1+\eps)^2$-approximates the persistence module 
$(H(\B_\alpha))_{\alpha\ge0}$.
\end{theorem}

Note that we obtain $(1+\eps)^2$ instead of $(1+\eps)$ because
we consider the continuous filtration instead of the discrete
filtration of $|\pixels_\alpha|$.

The theorem also implies that $(H(|\pixels_{\alpha}|))_{\alpha> 0}$
is a $(1+\eps)^2$-approximation of the \Cech filtration since
the \Cech filtration is dual to $(\B_\alpha)_{\alpha\ge0}$
and has the same persistence diagram.

\subsection{Connecting the scales}
\label{subsection:tower}

We now turn our attention to the approximation complex $\complex_\alpha$.
While at each scale $\alpha$, $\complex_\alpha$ is the nerve of $\pixels_\alpha$ and $|\pixels_\beta|$ forms a filtration, it is not
true that $\complex_\alpha\subseteq \complex_\beta$ for all $\alpha\le \beta$
since $\pixels_\alpha \not\subset \pixels_\beta$.
Therefore, it is not sufficient to apply the persistent nerve lemma 
of~\cite{co-pnl} directly, but it requires a more involved analysis, which 
we describe next.
We show that the complexes are connected by simplicial maps.

A first useful property is that $\complex_\alpha$ is a flag complex.
This follows from the following statement, which we prove in more
general form for later use. We define an axis-aligned cuboid
in $\R^d$ as the Cartesian product of $d$ intervals 
$I_1 \times \ldots \times I_d$ (where the degenerate case is allowed 
that $I_j$ consists of only one point).
For instance, all pixels of any lattice $L_\beta$ are
(non-degenerate) cuboids.

\begin{lemma}
\label{lemma:digital_cubeflag}
The nerve complex of any finite collection of cuboids is a flag complex.
\end{lemma}

\begin{proof}
We show that if there is a set of $(k+1)$ cuboids 
that pairwise intersect,
then all of them have a common intersection. 
The intersection of the set of $(k+1)$ cuboids 
\[
\left(I^{(0)}_1\times\ldots\times I^{(0)}_d\right),
\ldots,
\left(I^{(k)}_1\times\ldots\times I^{(k)}_d\right)
\]
is the cuboid 
\[\left(I^{(0)}_1\cap\ldots\cap I^{(k)}_1\right)
\times\ldots\times
\left(I^{(0)}_d\cap\ldots\cap I^{(k)}_d\right),\]
and it suffices to show that these intersection are non-empty
coordinate-wise. By assumption, $I^{(\ell_1)}_1\cap I^{(\ell_2)}_1$ 
is non-empty for all $0\leq \ell_1,\ell_2\leq k$. 
Helly's theorem~\cite{helly} for
the case of intervals implies that all $I^{(\cdot)}_1$ intersect commonly.
\end{proof}

Next, we consider $\alpha\in I$ and define a simplicial map 
$g:\complex_\alpha\to\complex_{(1+\eps)\alpha}$. This is simple
if $\complex_\alpha$ and $\complex_{(1+\eps)\alpha}$ 
are constructed using the same grid,
as in that case $\pixels_\alpha\subseteq\pixels_{(1+\eps)\alpha}$
and consequently, $\complex_\alpha\subseteq \complex_{(1+\eps)\alpha}$.
If the grid changes, there is no direct inclusion. There is, however,
a natural map $g'$ mapping each pixel $C$ in $\pixels_\alpha$ to the 
unique pixel $g'(C)$ in $\pixels_{(1+\eps)\alpha}$ that contains $C$.
That fact that $g'(C)$ is indeed in $\pixels_{(1+\eps)\alpha}$
follows from Sandwich Lemma.

\begin{lemma}
\label{lemma:digital_gsimplicial}
The map $g'$ extends to a simplicial map 
$g:\complex_{\alpha}\rightarrow\complex_{(1+\eps)\alpha}$.
\end{lemma}

\begin{proof}
From Lemma~\ref{lemma:digital_cubeflag} we know that $\complex$ is a flag complex.
Therefore it suffices to show that $g$ 
maps every edge of $\complex_{\alpha}$ 
to either a single vertex or an edge of $\complex_{(1+\eps)\alpha}$.
But that follows at once, because an edge in $\complex_{\alpha}$
corresponds to two pixels on the grid of $\complex_{\alpha}$
which are intersecting, and the map $g'$ as defined above either maps
both of them to the same pixel, or to two pixels which are also intersecting.
\end{proof}

By composing the maps above, we obtain maps 
$g^{\alpha_1,\alpha_2}:\complex_{\alpha_1} \rightarrow \complex_{\alpha_2}$
for all $0<\alpha_1\le\alpha_2 \in I$.
Using the standard filling technique, we define the (continuous) 
simplicial tower
\[
(\complex_{\alpha})_{\alpha>0}.
\]

\paragraph{Interleaving}
We next establish a relationship between the approximation tower 
$(\complex_{\alpha})_{\alpha> 0}$
and $(|\pixels_{\alpha}|)_{\alpha> 0}$.
More precisely, we will show that both towers yield
the same persistence diagram. 
Using Theorem~\ref{theorem:digital_basic_intlv},
this implies that $(\complex_{\alpha})_{\alpha> 0}$
is a $(1+\eps)^2$-approximation of the \Cech complex.

Since $\pixels_\alpha$ consists of convex objects, 
the nerve theorem (Theorem~\ref{theorem:nervethm}) asserts that 
$\pixels_\alpha \homt\complex_\alpha$, for each scale $\alpha$,
so they have isomorphic homology groups.
All that is left to show is that the homology map $g^\ast$ 
induced by the simplicial maps $g$ from above
commutes with the isomorphisms from the Nerve theorem.
We prove this indirectly by introducing an intermediate filtration
that is equivalent to both filtrations. This intermediate filtration
considers the union of all pixels $\pixels_\beta$ with $\beta\leq\alpha$.
We shift the technical details to 
Appendix~\ref{subsection:appendix-section3-interleaving-alternate} 
and just state the final result.

\begin{lemma}
\label{lemma:pixel_nerve_iso}
The persistence modules $(H(\complex_{\alpha}))_{\alpha> 0}$ 
and $(H(|\pixels_{\alpha}|))_{\alpha> 0}$ are isomorphic.
\end{lemma}

We conclude with the main result of this section.
Setting $\eps'=\eps/4$ in our approximation, 
we see that $(1+\frac{\eps}{4})^{2}<1+\eps$. 
We conclude that

\begin{theorem}
\label{theorem:digital_intlv_main}
$(H(\complex_{\alpha}))_{\alpha > 0}$ and 
$(H(\cech_{\alpha}))_{\alpha\ge 0}$ are 
$(1+\eps)$-approximations of each other, for $\eps\in (0,\frac{1}{20}]$.
\end{theorem}

\subsection{Size and Computation}
\label{subsection:sizecomputation}

\begin{theorem}
\label{theorem:digital_size}
For every $\alpha\geq 0$, 
the $k$-skeleton of the approximation tower $\complex_\alpha$ has size 
\[
n\left(\frac{1}{\eps}\right)^d2^{O(d\log d + dk)}.
\]
\end{theorem}

\begin{proof}
Let $\beta\leq\alpha$ be such that the complex $\complex_{\alpha}$ 
is built using the lattice $L_\beta$. Note that $\beta\geq \alpha/2$.
The sidelength of a pixel of this lattice is $\frac{\eps\beta}{4\sqrt{d}}$.
A ball of half this radius is contained inside the pixel.
Since the pixels are interior-disjoint, 
a simple packing argument shows that each ball in $\B_{\alpha}$ is 
covered by no more than 
\[
\left(\frac{\alpha}{\frac{\eps\beta}{8\sqrt{d}}}\right)^d
\le 
\left(\frac{16\sqrt{d}}{\eps}\right)^d
=\left(\frac{1}{\eps}\right)^d 2^{O(d\log d)}
\] 
pixels. 
There are $n$ balls of $\B_{\alpha}$,
hence, $\pixels_\alpha$ contains at most 
$n\left(\frac{1}{\eps}\right)^d 2^{O(d\log d)}$ pixels.
Each pixel is incident to at most $3^d-1=2^{O(d)}$ other pixels.
Each simplex incident to a pixel has vertices among these $2^{O(d)}$ pixels.
Therefore, the size of the $k$-skeleton incident to each pixel is then $2^{O(dk)}$.
In total, the $k$-skeleton of $\complex_\alpha$ has size
\[
n\left(\frac{1}{\eps}\right)^d 2^{O(d\log d)}2^{O(dk)}
=n\left(\frac{1}{\eps}\right)^d 2^{O(d\log d+dk)},
\]
independent of $\alpha$.
\end{proof}

To compute the complex at a given scale $\alpha\in I$, we first find the
pixels at that scale using a simple flooding algorithm that starts at vertices
of $P$.
Then we inspect the neighborhood of each pixel to compute the $k$-skeleton
incident to that vertex.
See Appendix~\ref{subsection:appendix-section3-algo}  
for more details.

\begin{theorem}
\label{theorem:digital_computation}
At each scale of $I$, the $k$-skeleton of the approximation tower 
and the simplicial map can be computed in time
\[
n\left(\frac{1}{\eps}\right)^d2^{O(d\log d +dk)}.
\]
\end{theorem}

\section{Removing the spread}
\label{section:spreadremoval}

The size and computation bounds from Section~\ref{subsection:sizecomputation}
only hold for a single considered scale. To get a bound on the total
size and complexity, 
one needs to multiply with the considered number of scales.
If the entire filtration is of interest, this number can be upper bounded
by the logarithm of the spread of the point set (the ratio of the diameter
and the closest distance). In this section, we remove that dependence
on the spread.

\subsection{A short overview of the approximation scheme}
\label{subsection:summary_spreadremoval}

We first informally illustrate the idea behind eliminating the 
dependence on spread.
The main exposition starts from Subsection~\ref{subsection:lazyunion}, 
where we start with the details of our construction.

The crucial observation guiding our improved approximation scheme is that
while the union of balls $\B_\alpha$ changes continuously for all 
$\alpha\ge 0$, the \Cech complex changes only at a finite number of scales.
For instance, if an edge $(p,q)$ enters the \Cech filtration at $\alpha$, 
then the intersection of the $\alpha$-balls at $p$ and $q$ captures the edge.
But at higher scales, it is pointless to allow these balls to grow further 
to represent this edge.
So we allow the balls to grow individually only in a lazy fashion, only
when they form new connections. 
This ensures that the set of edges of the original \Cech filtration 
are correctly captured.
Allowing for some slack in growing the balls, that is, if 
we allow the balls at $p,q$ to grow in the interval $[\alpha,2\alpha]$,
we can capture all simplices incident to $p$ and $q$.
Naturally, this lazy union of balls is homotopy-equivalent 
to the original union.
See Figure~\ref{figure:summary} for an example.

\begin{figure}[ht]
\centering
\begin{subfigure}[t]{0.33\columnwidth}
\centering
\includegraphics[width=0.8\textwidth]{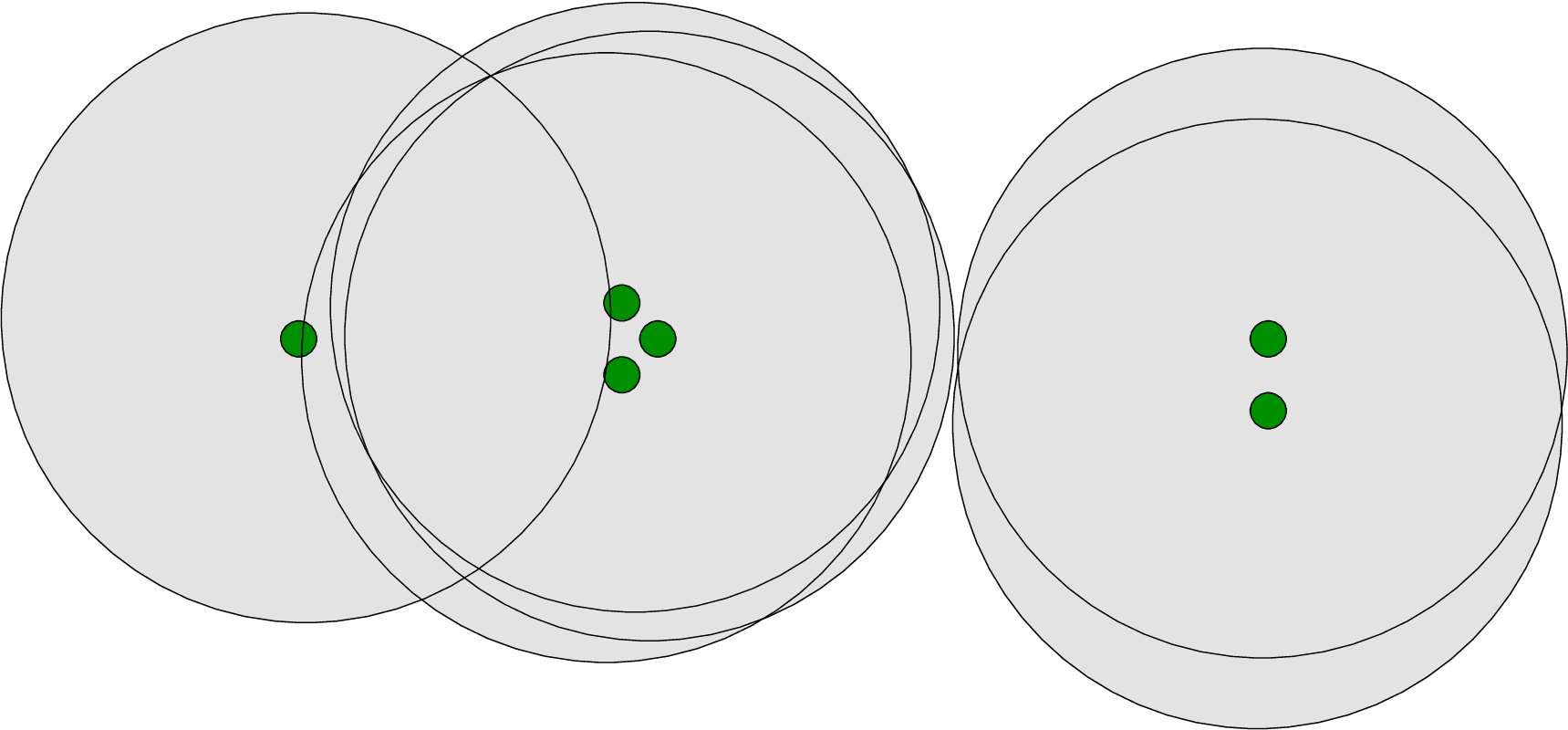}
\caption{
A union of balls.
}
\end{subfigure}%
\begin{subfigure}[t]{0.33\columnwidth}
\centering
\includegraphics[width=0.8\textwidth]{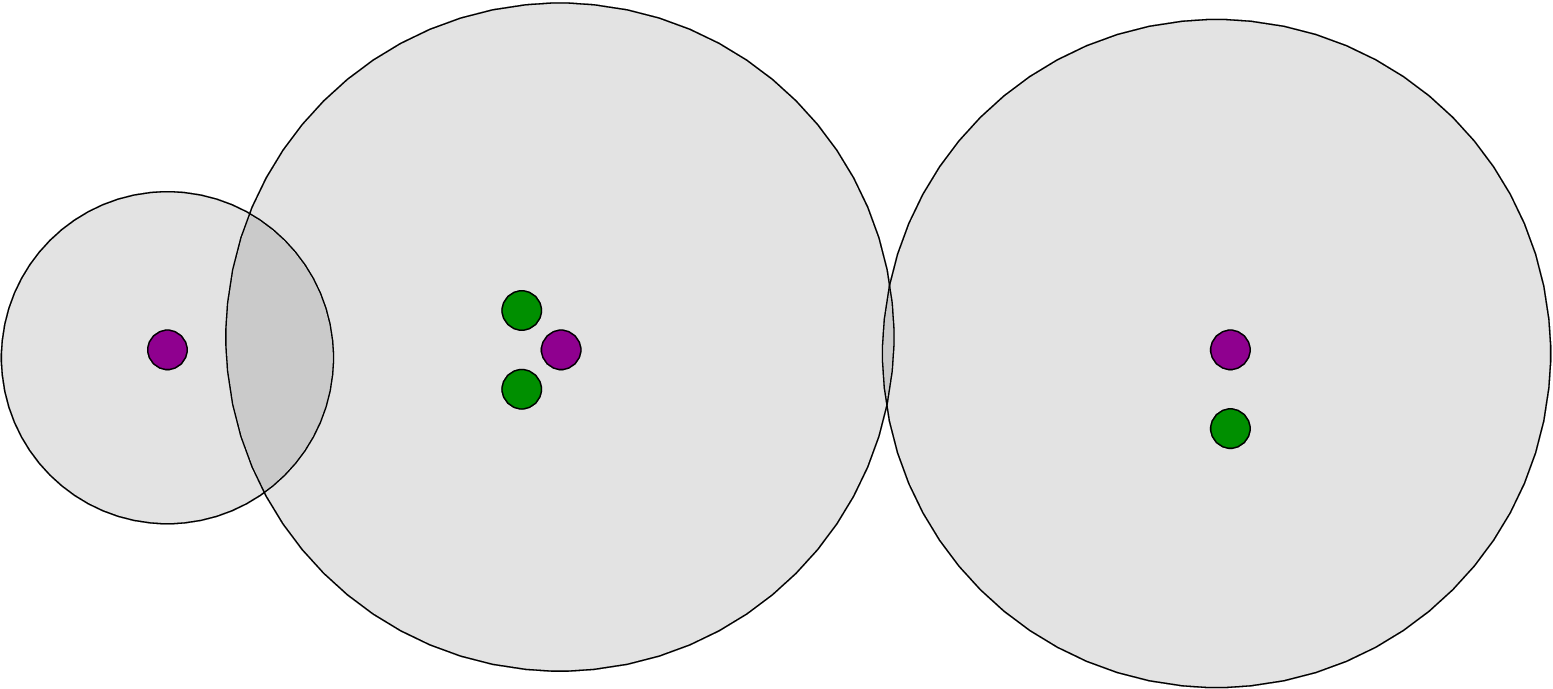}
\caption{
The corresponding lazy union of balls. The violet points are representatives.
}
\end{subfigure}
\begin{subfigure}[t]{0.33\columnwidth}
\centering
\includegraphics[width=0.8\textwidth]{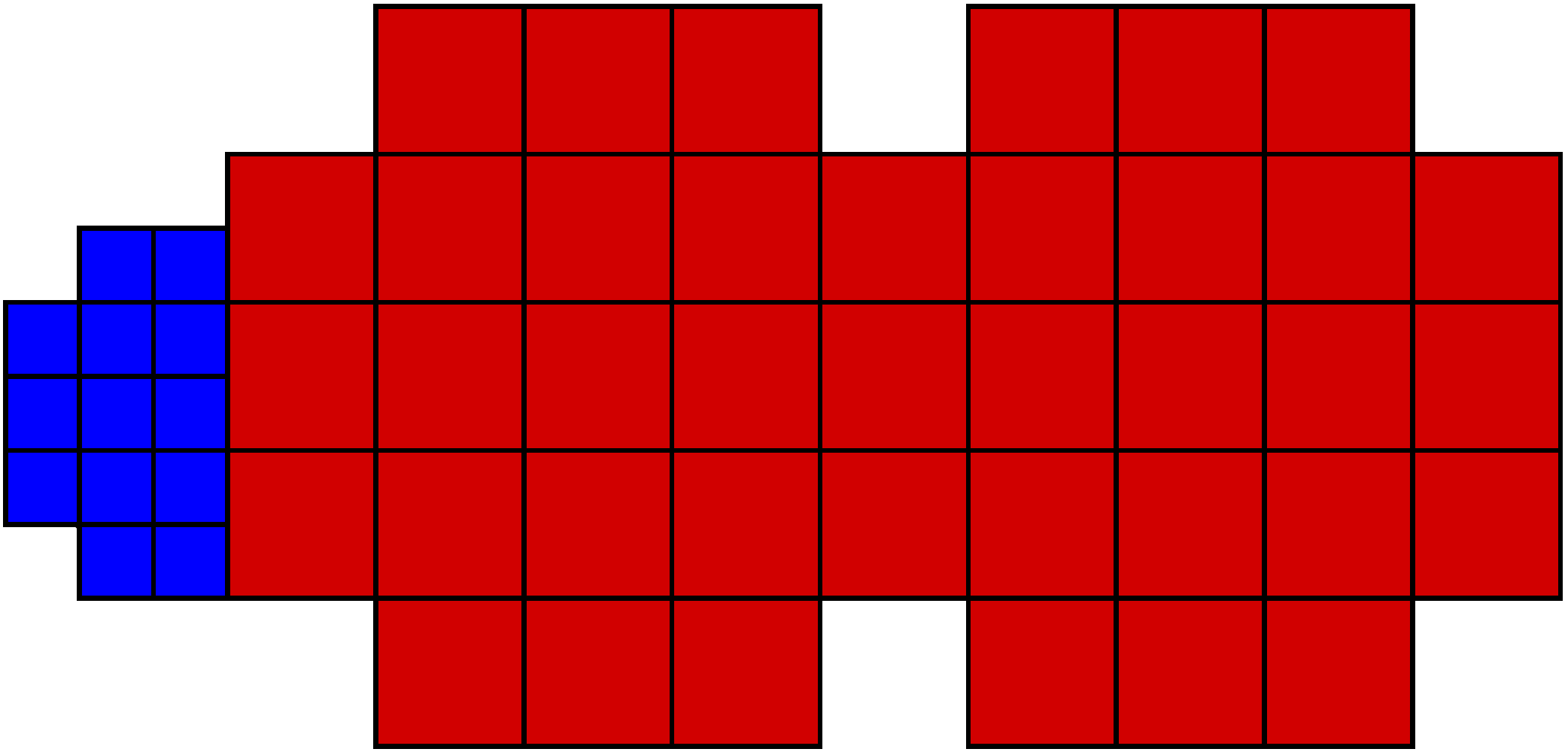}
\caption{
A pixelization.
}
\end{subfigure}
\caption{Original union, lazy union and the pixelization.}
\label{figure:summary}
\end{figure}

The lazy union is non-homogeneous, so it is non-trivial to pixelize it.
The ideal approach is to pixelize each ball according to its radius,
to get a non-homogeneous union of pixels.
As we show later, this does not pose a problem if the pixels are chosen 
as in Section~\ref{section:digitization}.
To construct the pixelization efficiently, 
we make use of a well-separated pair decomposition.
Each pair of points $(p,q)\in P$ has roughly the same distance as the distance
between any points of a well-separated pair that covers $(p,q)$.
Hence, the lazy balls at the representatives approximate the 
lazy balls at the points of the pair, and we use the former for pixelization.
The use of an $\eps$-WSPD ensures that there are at most $n(1/\eps)^{O(d)}$
lazy balls.
See Figure~\ref{figure:summary} for an example
(further details in Subsection~\ref{subsection:lazyunion}).

We collect the pixels that intersect the union of balls.
While the lazy union and its pixelization do not interleave on a space 
level as in Lemma~\ref{lemma:digital_sandwich}, we show that they still
form persistence modules that are closely interleaved 
(Subsection~\ref{subsection:pixelization}).
The nerves of these pixelizations are our approximation complexes,
and they can be connected to form a simplicial tower:
for this, we simply use the map that takes a pixel at a lower scale 
to the pixel at the higher scale that contains it.

To compute the pixelization, we first construct a WSPD and identify the 
intervals at which the representatives are growing. 
At each such scale, we identify only those balls that are expanding, 
pixelize them, and take the nerve.
This constructs our simplex inclusions at that scale.
More details follow in Subsection~\ref{subsection:sizecomputation_new}.

\subsection{Lazy union of balls}
\label{subsection:lazyunion}

We review the concept of a 
\emph{well-separated pair decomposition} (WSPD) from~\cite{ck-wspd},
which plays a crucial role in our algorithm.
A \emph{$\delta$-WSPD} of $P$ consists of pairs of the form 
$(A_i,B_i)\subset P\times P$ that satisfy
\begin{itemize}
\item each $(A_i,B_i)$ is a \emph{well-separated pair} (WSP), that means,
\[
diam(A_i),diam(B_i) \leq \delta\cdot\min_{p\in A_i,q\in B_i} \distance{p}{q},
\]
\item and for each pair of points $(p,q)\in P$, there exists a WSP
$(A_j,B_j)$ in the WSPD such that either $(p\in A_j,q\in B_j)$ 
or $(p\in B_j,q\in A_j)$. 
That means, a WSPD covers each pair of points of $P$.
\end{itemize}

Let $W$ denote a $\delta$-WSPD on $P$, where $\delta\le 1/10$.
For each WSP $(A,B)\in W$ let $\{P_A,P_B \}\subset P$ denote
the points of $P$ in $A$ and $B$, respectively.
We pick a point $rep_A\in P_A$ and call it the representative of $A$
(similarly for $B$).
We denote the distance between the representatives as 
$d(A,B):=\distance{rep_A}{rep_B}$.
From the WSPD-property, $max\{diam(A),diam(B)\}\le \delta d(A,B)$.
Using the triangle inequality,
we see that the distance between any two points of $P_A$ and $P_B$
lies in the interval $[d(A,B)/2,2d(A,B)]$ since $\delta\le 1/10$.
For reasons that becomes apparent later, we scale the interval by a factor
of $4$ and call
\[
R_{(A,B)}:=[d(A,B)/8,8d(A,B)]
\] 
the \emph{active interval} of the pair $(A,B)$.
For every point $p\in P$, we define

\begin{align*}
R_p:=\bigcup\{R_{(A,B)}\mid (A,B)\in W
\text{ and $p=rep_A$ or $p=rep_B$}\}\subset [0,\infty)
\end{align*}
as the \emph{active interval} of $p$. 
Each scale $\alpha\in R_p$ is an \emph{active scale}.
Next, we define for $p\in P$,
\[
r_{p}(\alpha):=\max\{r\leq\alpha\mid r\in R_p\}
\]
and set 
\[
\lazy_\alpha:=\bigcup_{p\in P} B(p,r_p(\alpha)).
\]

We can interpret these definitions as follows: $R_p$ specifies
a range of scales $\alpha$ for which the $\alpha$-ball of $p$
might encounter new intersections. The function $r_{p}(\alpha)$ is
monotonously increasing, $r_p(0)=0$ for all $p$, 
and $r_p(\alpha)=\alpha$ if $\alpha$ is an active
scale; otherwise, the radius of the ball just remains at the last 
encountered active scale. 
$\lazy_\alpha$ is the union of balls with radii given by the $r_p$ functions.

Note that $\lazy_\alpha\subseteq \lazy_{\alpha'}$ whenever $\alpha\leq\alpha'$.
Hence, $(\lazy_\alpha)_{\alpha\geq 0}$ is a filtration.

\begin{lemma}
\label{lem:lazy_interleaving}
The persistence module $(H(\lazy_\alpha))_{\alpha\geq 0}$
is a $(1+8\delta)$-approximation of  $(H(\B_\alpha))_{\alpha\geq 0}$ (and consequently, also of the \Cech filtration).
\end{lemma}

The (somewhat tedious) proof can be summarized as follows: we define an
intermediate filtration of balls where not only the balls of representatives,
but of all balls participating in a WSP grow. In the first part,
we show that this filtration can be sandwiched with 
that of $\lazy_\alpha$.
In the second part, we show that the intermediate filtration has the same
nerve as the union of $\alpha$-balls.

We now define the intermediate filtration. Let
\[
\tilde{R}_{(A,B)}:=[d(A,B)/4,4d(A,B)]
\] 
be a scaled version of the active interval from above. Define
\begin{align*}
\tilde{R}_p:=R_p\cup\bigcup\{\tilde{R}_{(A,B)}\mid (A,B)\in W
\text{ and $p\in A$ or $p\in B$}\}\subset [0,\infty),
\end{align*}

\[
\tilde{r}_{p}(\alpha):=\max\{r\leq\alpha\mid r\in \tilde{R}_p\},
\]
and set 
\[
\widetilde{\lazy}_\alpha:=\bigcup_{p\in P} B(p,\tilde{r}_p(\alpha)).
\]
Intuitively, all the balls centered at points in a WSP grow when the WSP 
is active (but those of the representatives grow for a slightly longer time).

\begin{lemma}
\label{lem:sandwich_tilde}
For all $\alpha\ge 0$, $\lazy_\alpha\subseteq\widetilde{\lazy}_\alpha\subseteq \lazy_{(1+8\delta)\alpha}$.
\end{lemma}
\begin{proof}
The first inclusion follows at once from the fact that
$r_p \leq \tilde{r}_p$ for all points of $P$ and all scales.
For the second inclusion, let $x$ be a point in $\widetilde{\lazy}_\alpha$
and let $p$ be such that $x\in B(p,\tilde{r}_p(\alpha))$. 
There is some WSP $(A,B)$ such that
$\tilde{r}_p(\alpha)$ lies in $R_{(A,B)}$ and $p$ is the representative of
$A$ or of $B$, or 
$\tilde{r}_p(\alpha)$ lies in some $\tilde{R}_{(A,B)}$ a WSP,
and $p\in A$ or $p\in B$. 
In the first case, it follows that $\tilde{r}_p(\alpha)=r_p(\alpha)$
and hence $x\in\lazy_\alpha\subseteq\lazy_{(1+8\delta)\alpha}$.
For the rest of the proof, we assume without loss of generality that 
$\tilde{r}_p(\alpha)$ lies in some $\tilde{R}_{(A,B)}$ and $p\in A$.

Let $q\neq p$ be the representative of $A$.
Then $R_{(A,B)}$ is a subset of $R_q$. By our choice of
$R_{(A,B)}$ and $\tilde{R}_{(A,B)}$, we know that $R_{(A,B)}$ contains
the value $(1+8\delta)\tilde{r}_p(\alpha)$ because $(1+8\delta)<2$
for $\delta<1/10$. On the other hand,
$\distance{p}{q}\leq\delta d(A,B)\leq 8\delta\tilde{r}_p(\alpha)$, where the last inequality comes from the fact
 $\tilde{r}_p(\alpha)\in R_{(A,B)}$. By triangle inequality,
$\distance{x}{q}\leq (1+8\delta)\tilde{r}_p(\alpha)$,
which implies the claim.
\end{proof}

The lemma implies that the two filtrations $(1+8\delta)$-approximate
each other. Next, we consider
\[
\widetilde{B}_\alpha:=\{ B(p,\tilde{r}_p(\alpha))\mid p\in P\}.
\]
which is the collection of balls whose union forms $\widetilde{\lazy}_\alpha$.
By the nerve theorem, the nerve of $\widetilde{B}_\alpha$ is homotopically
equivalent to $\widetilde{\lazy}_\alpha$, and since the balls have
non-decreasing radius, the persistence modules
$(H(\widetilde{\lazy}_\alpha))_{\alpha\geq 0}$
and $(H(\nerve(\widetilde{B}_\alpha))_{\alpha\geq 0}$ are persistence-equivalent.
The latter, however, is exactly the \Cech persistence module, as we show next.

\begin{lemma}
\label{lem:nerves_are_same}
$\nerve(\widetilde{B}_\alpha) = \nerve(\B_\alpha)=\cech_\alpha$
for all scales $\alpha\ge 0$.
\end{lemma}
\begin{proof}
Because $\tilde{r}_p(\alpha)\leq\alpha$,
$\nerve(\widetilde{B}_\alpha) \subseteq \nerve(\B_\alpha)$
follows at once.
For the other direction, fix any simplex $\sigma=(p_0,\ldots,p_k)$ 
of $\cech_\alpha$.
Let $\rho\leq\alpha$ denote the smallest radius 
such that the $\rho$-balls around the $p_i$
intersect. It suffices to show that $\tilde{r}_{p_i}(\alpha)\geq \rho$
for all $i=\{0,\ldots,d\}$.
For $p_i$ fixed, let $p_j$ denote the point of $\sigma$ with maximal
distance to $p_i$.
Clearly, $\distance{p_i}{p_j}\geq \rho$. On the other hand,
there exists a WSP $(A,B)$ that covers $(p_i,p_j)$ giving rise to
the interval $\tilde{R}_{(A,B)}=[d(A,B)/4,4d(A,B)]$.
With the well-separation property and the triangle inequality,
\begin{align*}
d(A,B) \geq  \distance{p_i}{p_j}-diam (A)-diam (B)
 \geq \rho-2\delta d(A,B),
\end{align*}
which leads to the inequality
\[
\rho\leq (1+2\delta)d(A,B)\leq 2d(A,B)
\]
since $\delta\leq 1/10$. Moreover, $\distance{p_i}{p_j}\leq 2\rho$, and
a similar calculation shows that
\[
\rho\geq \frac{(1-2\delta)}{2}d(A,B) \geq d(A,B)/4
\]
which implies that $\rho\in \tilde{R}_{(A,B)}$. Since $p_i\in P_A$ 
or $p_i\in P_B$,
this implies immediately that $\tilde{r}_{p_i}(x)\geq \rho$ for $x\geq \rho$.
\end{proof}

Combining Lemma~\ref{lem:nerves_are_same} with Lemma~\ref{lem:sandwich_tilde}
completes the proof of Lemma~\ref{lem:lazy_interleaving}.

\subsection{Pixelization}
\label{subsection:pixelization}

We next define a collection of pixels $\lpixels_\alpha$
for each $\alpha\in I$. 
In contrast to Section~\ref{section:digitization},
these pixels will be taken from different grids. 
For technical reasons, we enlarge the critical intervals $R_{(A,B)}$
by the smallest amount such that their endpoints are points in $I$. 
The definitions of $R_p$, $r_p(\alpha)$ and $\lazy_\alpha$ get adapted
accordingly, without affecting the claims of Lemma~\ref{lem:sandwich_tilde}
and Lemma~\ref{lem:nerves_are_same}.

We assume again without loss of generality that we start the construction 
at a minimal scale $\alpha_0$
that is sufficiently small and set $\lpixels_{\alpha_0}=P$.
We define $\lpixels_{\alpha}$ for $\alpha>\alpha_0$ 
inductively:
let $\beta$ denote the largest power
of $2$ that is not larger than $\alpha$ 
as in Section~\ref{section:digitization}.
Call a point $p\in P$ \emph{active} if $r_p(\alpha)=\alpha$,
and \emph{inactive} otherwise. Let $\lpixels'_\alpha$ be the collection
of all pixels whose centers lie in the any ball of radius $(1+\eps/2)r_p(\alpha)$ at an active points $p$.
Define $\lpixels_\alpha$ as the inclusion-maximal pixels of the set
$\lpixels_{\alpha/(1+\eps)}\cup \lpixels'_\alpha$.
We denote by $|\lpixels_\alpha|$ the union of the pixels in $\lpixels_\alpha$.
See Figure~\ref{figure:hetero_2d} for an illustration.
We use the standard filling technique to define $|\lpixels_\alpha|$
for every $\alpha\geq 0$.

\begin{figure}[ht!]
\centering
\begin{subfigure}[t]{0.4\columnwidth}
\centering
\includegraphics[width=0.99\textwidth]{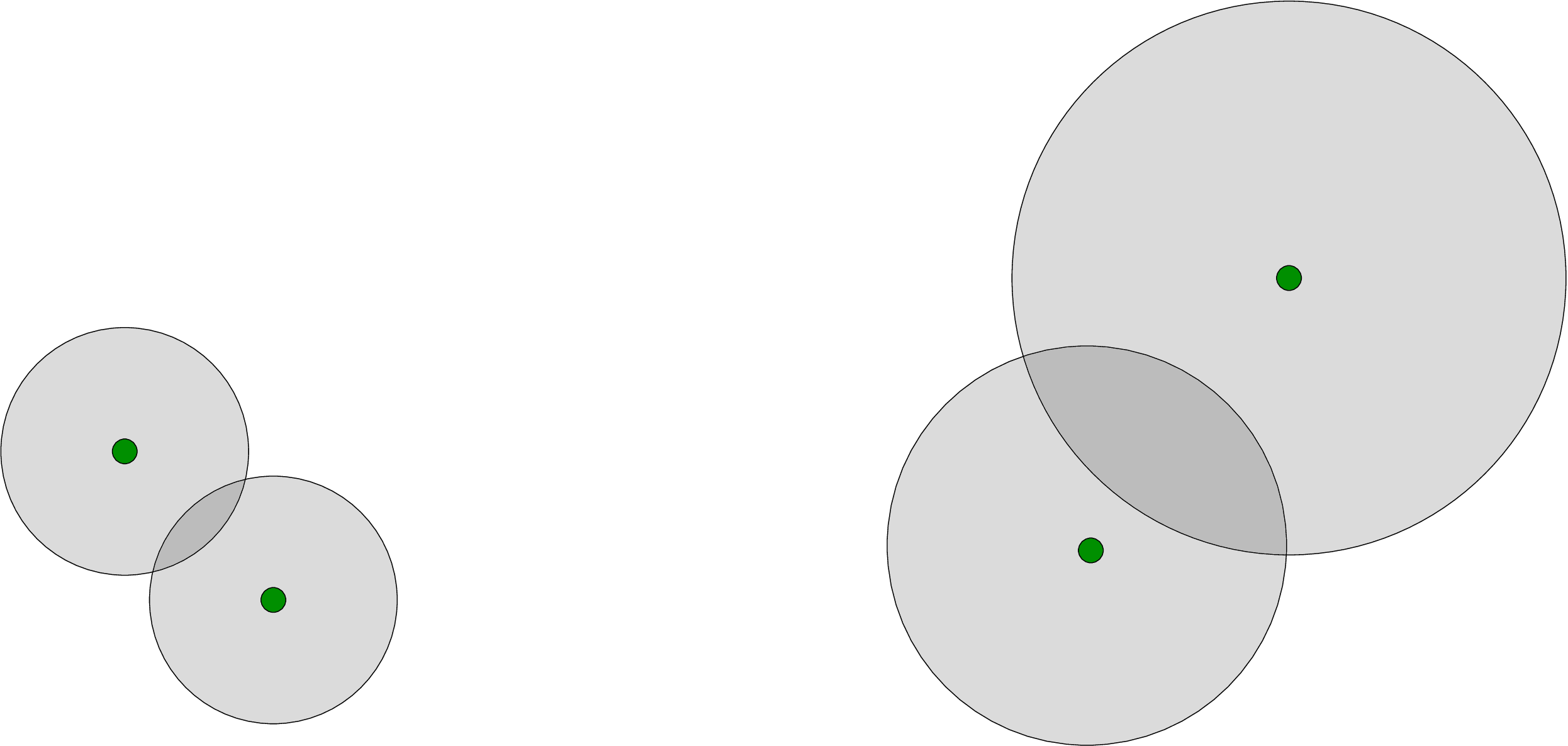}
\caption{
A collection of balls of three different radii.
}
\end{subfigure}\hspace{0.5em}
\begin{subfigure}[t]{0.4\columnwidth}
\centering
\includegraphics[width=0.99\textwidth]{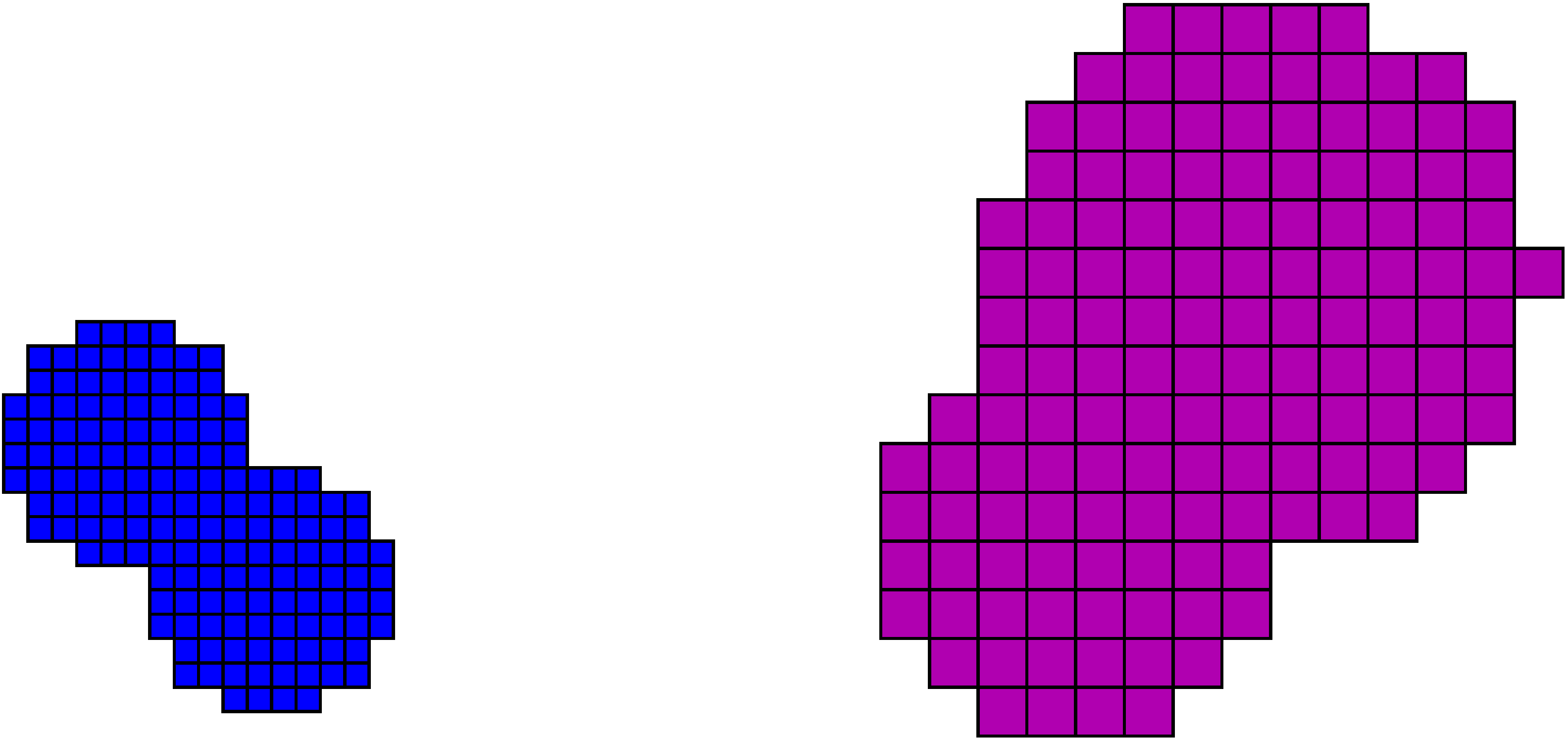}
\caption{
The corresponding pixels of two sizes.
}
\end{subfigure}
\caption{An example of digitization using lazy pixels.}
\label{figure:hetero_2d}
\end{figure}

To rephrase our construction: at scale $\alpha$, we resample the pixels
of all balls which are growing from $\alpha/(1+\eps)$ to $\alpha$
in the filtration $(\lazy_\alpha)_{\alpha\ge 0}$, 
and add them to the set of pixels,
removing lower-scale pixels which are covered by the new ones.

The idea of the digitization is to approximate $\lazy_\alpha$ with pixels.
Indeed, we have that $\lazy_\alpha\subseteq |\lpixels_\alpha|$. That follows
inductively, observing that the balls that grow from $\lazy_\alpha$
to $\lazy_{(1+\eps)\alpha}$ are covered by pixels.
It is also true that $|\lpixels_\alpha|\subseteq \lazy_\alpha^{1+\eps}$,
where $\lazy_\alpha^{1+\eps}$ is obtained from $\lazy_\alpha$ 
by enlarging every ball by a factor of $(1+\eps)$. 
The proofs are similar to the proof of Lemma~\ref{lemma:digital_sandwich}.
However, it is not true that $\lazy_\alpha^{1+\eps}=\lazy_{(1+\eps)\alpha}$
in general (more precisely, neither $\lazy_\alpha^{1+\eps}\subseteq\lazy_{(1+\eps)\alpha}$ nor $\lazy_\alpha^{1+\eps}\supseteq\lazy_{(1+\eps)\alpha}$ holds in general). Therefore, the simple sandwiching strategy 
from Section~\ref{section:digitization} fails. We need a more refined argument,
showing that the interleaving still works on the homology level.

For that, the following observation is crucial:
\begin{lemma}
The spaces $\lazy_\alpha^{1+\eps}$ and $\lazy_{(1+\eps)\alpha}$
are homotopically equivalent, and the maps $f_1,f_2$ realizing this
homotopic equivalence can be chosen 
as deformation retracts
to the intersection $\lazy_\alpha^{1+\eps}\cap\lazy_{(1+\eps)\alpha}$.
\end{lemma}

The statement follows from the following statement on unions of balls,
noting that the balls forming $\lazy_\alpha^{1+\eps}$, $\lazy_{(1+\eps)\alpha}$, and $\lazy_\alpha^{1+\eps}\cap\lazy_{(1+\eps)\alpha}$
all have the same nerve by construction.

\begin{lemma}
Let $\mathcal{A}:=\{A_1,\ldots,A_n\}$, 
$\mathcal{B}:=\{B_1,\ldots,B_n\}$ be collections of closed balls
such that $A_i\subseteq B_i$ and $A_i$ and $B_i$ have the same center
for all $i=1,\ldots,n$. 
If $\nrv(\mathcal{A})=\nrv(\mathcal{B})$, there
is strong deformation retraction from $|\mathcal{B}|$ to $|\mathcal{A}|$.
\end{lemma}

We only briefly sketch the proof idea and postpone details to an extended
version.
Note that because the balls are closed, there exists an $\eps>0$,
such that we can increase all balls of $\mathcal{B}$ without changing
the nerve.
This yields a collection of balls $\mathcal{C}:=\{C_1,\ldots,C_n\}$ 
with the same centers, such that $A_i\subsetneq C_i$.
We can define a collection of distance-like functions $f_1,\ldots,f_n$,
such that $A_i$ is the sublevel set of $f_i$ for value $1$,
and $C_i$ is the sublevel set of $f_i$ for value $2$.
Let $f$ denote the lower envelope of these functions. This function
is continuous, but not differentiable; however, using techniques
similar to the \emph{generalized gradient}~\cite{ccl-sampling}, 
it is possible to define a flow on $|\mathcal{C}|$.
The singularities of this flow are outside of 
$|\mathcal{C}|\setminus|\mathcal{A}|$ because a singular points
of the generalized gradient triggers a change in the nerve.
That flow hence defines a deformation retract from $|\mathcal{C}|$
to $|\mathcal{A}|$.

We can use the above to define a deformation retract from $|\mathcal{B}|$
to $|\mathcal{A}|$: 
Simply restricting
the retract to $|\mathcal{B}|$ is not sufficient 
because a point of $|\mathcal{B}|$ may be pushed out of $|\mathcal{B}|$
during the flow.
However, note that $\mathcal{C}$ was defined using
some parameter $\eps$.
Call a point $x\in|\mathcal{B}$ \emph{$\eps$-affected} if there exists
an $i$ such that $x\notin B_i$, but $x\in C_i$.
For every $x\in |\mathcal{B}|$,
we can find some $\eps>0$ such that an open neighborhood of $x$
is not $\eps$-affected. If $x$ is not $\eps$-affected,
its flow does not change when further decreasing $\eps$.
Hence, for every $x\in |\mathcal{B}|$, we can define its flow to
be the limiting flow for $\eps\to 0$. Because of the neighborhood property
above, the resulting flow is continuous on $|\mathcal{B}|$,
and defines the desired deformation retraction.

\begin{lemma}[Lazy sandwich lemma]
\label{lemma:lazysandwich_final}
The two persistence modules $(H(|\lpixels_\alpha|))_{\alpha\ge 0}$ and 
$(H(\lazy_\alpha))_{\alpha\ge 0}$ are $(1+\eps)^2$-approximations 
of each other.
\end{lemma}

\begin{proof}
We need to define interleaving maps such that the four diagrams of ~(\ref{diag:strong_diag}) commute. 
We restrict our attention to scales in $I$ and show
a $(1+\eps)$-interleaving for these scales, which implies the result.

From $H(\lazy_\alpha)\to H(|\lpixels_{(1+\eps)\alpha}|)$, we take the homology map
induced by the inclusion $\lazy_\alpha\subseteq |\lpixels_\alpha|$. With that,
the diagram
\begin{equation}
\xymatrix{
H(\lazy_{\alpha}) \ar[rr] \ar[rd] && H(\lazy_{\alpha'}) \ar[rd] \\
& H(|\lpixels_{(1+\eps)\alpha}|) \ar[rr] & & H(|\lpixels_{(1+\eps)\alpha'}|)
}
\end{equation}
commutes, since it already commutes on a space level.

We define the map $\phi:H(|\lpixels_{\alpha}|)\to H(\lazy_{(1+\eps)\alpha})$
as the composition of the homology map 
$H(|\lpixels_{\alpha}|)\to H(\lazy_\alpha^{1+\eps})$ induced by inclusion 
and the map $f^\ast:H(\lazy_\alpha^{1+\eps})\to H(\lazy_{(1+\eps)\alpha})$,
where $f$ is the map from above realizing the homotopy equivalence. 
We consider the diagram
\begin{equation}
\xymatrix{
 & H(\lazy_{c\alpha})\ar[r] & H(\lazy_{c\alpha'})
\ar[rd] &\\
H(|\lpixels_{\alpha}|) \ar[rrr] \ar[ru]^{\phi} &&& H(|\lpixels_{c^2\alpha'|})
}
\end{equation}
where $c=(1+\eps)$ and all maps except $\phi$ are induced by inclusion.
Now, fixing a (singular) $p$-cycle $z$ over $|\lpixels_{\alpha}|$,
composition with $f$ yields a $p$-cycle $f(z)$ in $H(\lazy_{(1+\eps)\alpha})$
which in turn includes as a $p$-cycle in $|\lpixels_{(1+\eps)^2\alpha'}|$.
Let $f'$ denote the map $\lazy_{(1+\eps)\alpha}\to \lazy_\alpha^{1+\eps}$
in the opposite direction of the homotopy equivalence. Since
$f'$ is realized as a deformation retraction to the intersection, and
$f(z)$ lies in the intersection, it follows that $f'(f(z))=f(z)$.
On the other hand, $f'\circ f$ is homotopic to the identity,
hence, $f(z)$ is homotopic to $z$. This proves that the $p$-cycles
$z$ and $f(z)$ represent the same homology class, proving that the above
diagram commutes. We skip the remaining two diagrams, which can be handled
with similar methods.
\end{proof}

\subsection{Approximation tower and equivalence to the lazy pixels filtration}

We define our approximation complex as 
$\lcomplex_\alpha := \nerve(\lpixels_\alpha)$.
Since $\lpixels_\alpha$ is a collection of pixels, $\lcomplex_\alpha$
is a flag complex using Lemma~\ref{lemma:digital_cubeflag}.
In general, $\lcomplex_\alpha$ does not include into 
$\lcomplex_{(1+\eps)\alpha}$ because pixels might be 
removed when passing to a larger grid.

We define a map $g':\lpixels_\alpha\to\lpixels_{(1+\eps)\alpha}$ that maps 
a pixel of $\lpixels_\alpha$ 
to the unique pixel in $\lpixels_{(1+\eps)\alpha}$ that contains it.
Because $\lcomplex_\alpha$ is a flag complex and intersecting pixels
are mapped to intersecting pixel under $g'$, it follows directly
that $g'$ induces a simplicial map 
$g:\lcomplex_\alpha\to\lcomplex_{(1+\eps)\alpha}$.
Using the same proof strategy as in Appendix~\ref{subsection:appendix-section3-interleaving-alternate}, we can show that the towers $(\lcomplex_\alpha)_{\alpha\geq 0}$
and $(|\lpixels|_\alpha)_{\alpha\geq 0}$ are persistence-equivalent.

We summarize the results of the section so far.

\begin{theorem}
Using a $\frac{\eps}{8}$-WSPD, our construction yields 
a persistence module $(H(\lcomplex_\alpha))_{\alpha\ge 0}$ 
that is a $(1+\eps)^3$-approximation of the \Cech persistence module.
\end{theorem}

\begin{proof}
$(H(\lpixels_\alpha))_{\alpha\ge 0}$ is a $(1+\eps)^2$-approximation
of $(H(\lazy_\alpha))_{\alpha\ge 0}$ from Lemma~\ref{lemma:lazysandwich_final}.
Further, the persistence module $(H(\lazy_\alpha))_{\alpha\ge 0}$ is a $(1+\eps)$-approximation
of the \Cech persistence module from Lemma~\ref{lem:lazy_interleaving}.
\end{proof}

We can easily get a $(1+\eps)$-interleaving as well by using $\eps'=\eps/4$
in our approximation, since 
$\left(1+\eps'\right)^3=\left(1+\frac{\eps}{4}\right)^3<1+\eps$.

\subsection{Size and computation time}
\label{subsection:sizecomputation_new}

Since our construction heavily depends on WSPDs, we state the well-known
complexity bound here:

\begin{theorem}[\cite{ck-wspd,hp-book}]
Given $\delta\in (0,1]$, a $\delta$-WSPD  of size 
$n \left(\frac{1}{\delta}\right)^{O(d)}$
can be computed in time
\[n\log n 2^{O(d)} + n \left(\frac{1}{\delta}\right)^{O(d)}.\]
\end{theorem}

Since we require an $(\eps/8)$-WSPD for our construction, the same
bounds hold up to a constant factor of the form $2^{O(d)}$
when replacing $\delta$ by $\eps$.

\begin{theorem}
\label{theorem:lazy_size}
The size of the $k$-skeleton of the tower $(\lcomplex_\alpha)_{\alpha\ge 0}$ 
is
\[
n\left(\frac{1}{\epsilon}\right)^{O(d)} 2^{O(d\log d + dk)}.
\]
\end{theorem}

\begin{proof}
Recall that the size of a tower is the number of simplices added in total
when increasing the scale.
We first bound the number of pixels that are added to the tower
by $n\left(\frac{d}{\epsilon}\right)^{O(d)}$.
This is done by a charging argument, where we charge the inclusion of a pixel
into the tower to one pair in the WSPD and show that each pair is charged
at most $\left(\frac{d}{\epsilon}\right)^{O(d)}$ times. The claim follows
then with the size bound of the WSPD.

When constructing $\lcomplex_\alpha$, we add new pixels to sample a ball
around $p\in P$ only if $r_p(\alpha)=\alpha$. That, in turn, means
that there is a WSP $(A,B)$ such that $\alpha\in R_{(A,B)}$
and $p$ is the representative of $A$
or of $B$. We charge the new pixels for the $\alpha$-ball around $p$
to the WSP $(A,B)$. Since $(A,B)$ has two representatives, it gets charged
only by the pixels of two balls. By an analogue packing argument as
in Theorem~\ref{theorem:digital_size}, 
an $\alpha$ ball gives rise to
not more than 
$\left(\frac{d}{\epsilon}\right)^d$
pixels, so $(A,B)$ gets charged by that many pixels per scale in the worst
case.

Finally, the pair $(A,B)$ can only get charged for scales that are
included in $R_{(A,B)}$. A simple calculation shows that the number of such
scales is the smallest $k$ that satisfies $(1+\eps)^k\geq 64$
which can be solved as $k=O(1/\eps)$. That means, multiplying
with number of scales on which $(A,B)$ can be charged does not change the
bound. This proves the bound on the number of pixels.

We next bound the number of simplex inclusions.
Again, we use a charging argument, charging the inclusion of a simplex
to one pixels in its boundary that has minimal size among all pixels
in the simplex. We show that every pixel gets charged at most $2^{O(dk)}$
times, which completes the proof.
A pixel $C$ can only have up to
$3^d-1$ neighbors that are of same or greater size than $C$. Hence,
at the scale $\alpha$ where it gets included, 
$C$ can only get charged for $3^{dk}=2^{O(dk)}$
simplices of dimension $\leq k$.
However, $C$ might also get charged at scales larger than $\alpha$.
This, in turn, can only happen if $C$ gets a new neighbor along a face 
on which it was not intersecting another pixel before (otherwise, the
simplex is not included, but originates from a previous scale).
Since this can only happen $3^d-1$ times, the total number of simplices
that $C$ gets charged for is still $2^{O(dk)}$.
\end{proof}

Next, we describe an algorithm to compute the tower.
By ``computing'' a tower, we mean the following: the output is a list
of tokens of three different types:

\begin{itemize}
\item ``(scale, $s$)'', denoting that the tower passes to scale~$s$
\item ``(add, $\sigma$)'', denoting that the simplex $\sigma$ is added to the tower
\item ``(contract, $v_0,v_1$)'', denoting that the vertices $v_0$ and $v_1$
are contracted in the tower. 
\end{itemize}

These three operations specify the tower completely, because every simplicial
map can be written as a sequence of inclusions and contractions~\cite{dfw-gic,ks-twr}.

A first step is to compute the $\frac{\eps}{8}$-WSPD $W$.
We traverse its pairs and compute the active ranges of each pair.
Let
\[
U\gets\bigcup_{(A,B)\in W} R_{(A,B)}.
\]

Clearly, the scales in $I$ that lie outside of $U$ can be disregarded
because no new pixels are added to $\lpixels_\alpha$. In total, the algorithm
considers only $O(1/\eps)$ scales per WSP, so $n\left(\frac{d}{\epsilon}\right)^{O(d)}$ scales in total (and these scales can be computed in the same
complexity).

For each scale, the algorithm determines the new pixels to be sampled.
On a scale $\alpha$, we can efficiently determine the points $p\in P$
where $r_p(\alpha)=\alpha$, for instance using a priority queue.
For each active point, we compute
the new pixels of its $\alpha$-ball, using the same flooding algorithm
as in Appendix~\ref{subsection:appendix-section3-algo}. The complexity
is proportional to the number of pixels added, and thus
the total complexity is bounded by $n\left(\frac{d}{\epsilon}\right)^{O(d)}$.

Next, the algorithm computes the contraction needed at scale $\alpha$.
Note that this problem simply means to find pixels in $\lpixels_{\alpha/(1+\eps)}$ which are covered by a pixel in $\lpixels_\alpha$. Each such pair of pixels
defines one contraction of the old pixel into the new one.
To find them efficiently, we store the collection of pixels
of $\lpixels_\alpha$ as leaves of a (compressed) quad-tree $Q$~\cite{hp-book}. 
Whenever a new pixel is added to $\lpixels_\alpha$, it is also added to $Q$
as a new node $v$.
All leaves in $Q$ that are children of $v$ are removed from the quad-tree,
and the corresponding contraction is enlisted.

Finally, the algorithm finds the $k$-simplices with $k\geq 1$
added to $\lcomplex_\alpha$. We only discuss the case $k=1$; the
simplices for higher $k$ can be combinatorially enlisted in time
proportional to the number of simplices because $\lcomplex_\alpha$ is
a flag complex~(see the end of Appendix~\ref{subsection:appendix-section3-algo}).
For edges, it suffices to iterate through the newly added vertices once more
and query the quad-tree data structure for all neighbors of the corresponding
leaf. It is not hard to see that all neighbors can be found in time
$O(2^d+N)$, where $N$ is the number of neighbors encountered.
Hence, the total time to find edges is also bounded
by $n\left(\frac{d}{\epsilon}\right)^{O(d)}$. That concludes the description
of the algorithm. In total, the running time
is dominated by the construction of the WSPD and the enumeration
step of all simplices, which is bounded by the size of the tower in
Theorem~\ref{theorem:lazy_size}. The space complexity
is dominated by the size of the output to be produced. 
We summarize the result.

\begin{theorem}
\label{theorem:lazy_time}
Computing the approximation tower takes time
\[
n \log n 2^{O(d)} + 
n\left(\frac{1}{\epsilon}\right)^{O(d)} 
2^{O(d\log d + dk)}, 
\]
and space
\[
n\left(\frac{1}{\epsilon}\right)^{O(d)} 2^{O(d\log d + dk)}.
\]
\end{theorem}

\section{Conclusion}
\label{section:conclusion}

We presented an approximation scheme for approximating \Cech filtrations
in Euclidean space.
The main idea is that by introducing additional sample points,
we can asymptotically reduce the size of the complex,
because the number of close-by points is reduced to a constant
independent of $n$ and $\eps$.

We briefly discuss the extensions of our results mentioned in the introduction.
While turning our tower into a (flag complex) filtration is mostly a matter of
technicalities, employing the permutahedral instead of the cubical grid
requires an additional bag of techniques. The major complication
comes from the fact that we lose the property of cubical grids that
one pixel on the lower scale is always completely contained
in a single pixel on the next scale. For instance, this property made
the construction of the simplicial maps connecting the scales natural
and easy to prove. For the permutahedral case, we need heavier
algebraic machinery, including barycentric subdivisions and acyclic carriers
to obtain similar results.

While the introduction of sample points yields improved theoretical guarantees,
the algorithm in its current form is impractical, because the number
of sample points is large even for well-conditioned problem instances.
We discuss two possibilities to reduce that number. First of all, 
instead of covering an $\alpha$-ball uniformly by pixels of a certain scale,
we could employ a quad-tree like subdivision and approximate the interior
of the ball with larger pixels in practice. Moreover, it might be
advisable to reduce the obtained complex further by applying elementary
collapses to it which do not change the homotopy type.
In very recent work~\cite{pbp-collapse}, it has been demonstrated how such
collapses can be performed efficiently in a simplicial complex.
While these extensions would certainly help in practice, the question
remains whether they can also lead to a further asymptotic improvements.

\ignore{
Our approximation schemes contain a factor of $2^{O(d\log d)}$ in our bounds.
This is the case even when we wish to compute medium to low-dimensional
persistent homology, that is, when $k\ll d$.
It would be an interesting exercise to see whether the pixelization
technique can be adapted to be more efficient while constructing $k$-skeletons.
}

\paragraph{Acknowledgements}
Aruni Choudhary is partly supported by ERC grant StG 757609.
Michael Kerber is supported by Austrian Science Fund (FWF) grant number P 29984-N35. 
Sharath Raghvendra acknowledges support of NSF CRII grant CCF-1464276.

\appendix
\label{appendix:usual}

\section{Missing proofs for Section~\ref{section:digitization}}
\label{section:appendix-section3}

\subsection{Proof of Lemma~\ref{lemma:pixel_nerve_iso}}
\label{subsection:appendix-section3-interleaving-alternate}

In this subsection we prove that the filtration of pixels $(|\pixels_\alpha|)_{\alpha>0}$ is persistence-equivalent to the tower of
nerves $(\complex_\alpha)_{\alpha>0}$.
For that, we introduce the intermediate object
\[
S'_\alpha:= \cup_{\beta\le \alpha} \pixels_\beta,
\]
where $\beta$ ranges over scales in $I$.
We can assume for simplicity that the range of
scales is finite, for instance by ignoring all scales which are smaller
than a third times the closest point distances among points in $P$,
and those scales that are larger than the diameter of $P$.
In that way, $S'_\alpha$ also becomes a finite object.

We define $Y_\alpha:=\nerve (S'_\alpha)$ for $\alpha\in I$,
and by the filling technique extend it to $(Y_\alpha)_{\alpha\geq 0}$.
This is a filtration, because $S'_\alpha\subseteq S'_{(1+\eps)\alpha}$.
By the Sandwich Lemma, we also have that $|S'_\alpha|=|S_\alpha|$.
Since the Nerve isomorphism on homology commutes with maps induced by inclusion~\cite{co-pnl},
we obtain directly that $(Y_\alpha)_{\alpha\geq 0}$ 
and $(|\pixels_\alpha|)_{\alpha>0}$ are persistence-equivalent.
In the rest of the section, we will show that also $(Y_\alpha)_{\alpha\geq 0}$ 
and $(\complex_\alpha)_{\alpha\geq 0}$ are persistence-equivalent.

Note that the Nerve Theorem implies that for each $\alpha$,
$\complex_\alpha$ and $Y_\alpha$ are homotopically equivalent, hence
their homology groups are isomorphic. We will construct such an isomorphism
explicitly. Recall from Section~\ref{section:digitization} that we have
constructed a simplicial map $g:\complex_\alpha\to\complex_{(1+\eps)\alpha}$
which maps a pixel of $\pixels_\alpha$ to the pixel of $\pixels_{(1+\eps)\alpha}$
that contains it. For $\beta\leq\alpha\in I$, we can define the map
$g_{\beta\to\alpha}$ as the composition of such $g$-maps,
yielding a simplicial map from $\complex_\beta$ to $\complex_{\alpha}$.

Next, we define a vertex map
$
g_{\to\alpha}':S'_\alpha\to S_\alpha
$
as follows: for a pixel $C$ in $S'_\alpha$, let $\beta\leq\alpha$ be such that
$C\in S_\beta$. Then, $g_{\to\alpha}'(C):=g_{\beta\to\alpha}(C)$. In words,
the map assigns to a pixel at scale $\leq\alpha$ the (unique)
pixel at scale $\alpha$ that contains it.

\begin{lemma}
The map $g_{\to\alpha}'$ induces a simplicial map $g_{\to\alpha}:Y_\alpha\to\complex_\alpha$.
\end{lemma}
\begin{proof}
Note that $Y_\alpha$ is a flag complex by Lemma~\ref{lemma:digital_cubeflag},
so it suffices to consider edges. 
The statement follows as in the proof of Lemma~\ref{lemma:digital_gsimplicial}
from the fact that if two pixels are intersecting, they either coincide
or still intersect when mapped to a larger scale.
\end{proof}

By the definition of the $g$-maps, we have that for $\beta\leq\alpha\in I$, 
\[
g_{\beta\to\alpha}\circ g_{\to\beta}=g_{\to\alpha}
\]
which implies at once that the following diagram commutes
\begin{eqnarray}
\label{equation:digital_intlv_alt}
\xymatrix{
\ldots\ar[r]&Y_\beta \ar[r]^{\inc} \ar[d]^{g_{\to\beta}} & 
Y_\alpha \ar[d]^{g_{\to\alpha}} \ar[r] 
& \ldots
\\
\ldots\ar[r]&\complex_\beta \ar[r]^{g_{\beta\to\alpha}} &
\complex_\alpha \ar[r]
& \ldots
}
\end{eqnarray}
To prove the equivalence of $(\complex_\alpha)$ and $(Y_\alpha)$,
it remains to show that the induced map 
\[
g_{\to\alpha}^\ast:H(Y_\alpha)\to H(\complex_\alpha)
\]
is an isomorphism. We do so by showing that the map 
$\inc^\ast:H(\complex_\alpha)\to H(Y_\alpha)$ induced by inclusion is inverse
to $g_{\to\alpha}^\ast$. Note that $g_{\to\alpha}$ is the identity on
the subcomplex $\complex_\alpha$, hence $g_{\to\alpha}\circ inc$ 
is the identity map on $\complex_\alpha$ as well. It follows that
$g_{\to\alpha}^\ast\circ inc^\ast$ is the identity map on the homology groups.

For the other direction, 
we make use of the following standard definition~\cite{munkres}:
two simplicial maps $f_1,f_2:K\rightarrow L$ are called \emph{contiguous}
if for every simplex $\sigma=(v_0,\ldots,v_k)\in K$,
the set of vertices $\{f_1(v_0),\ldots,f_1(v_k),f_2(v_0),\ldots,f_2(v_k)\}$
forms a simplex in $L$. Two contiguous maps are homotopic and therefore,
the induced maps $f_1^\ast$, $f_2^\ast$ on homology are equal.

\begin{lemma}
The maps $\id_{Y_\alpha}$ and
$\inc\circ g_{\to\alpha}:Y_\alpha\to Y_\alpha$
are contiguous.
\end{lemma}
\begin{proof}
Fix any simplex $\sigma=(x_0,\ldots,x_k)$ in $Y_\alpha$
and consider the set of vertices 
\[
(x_0,\ldots,x_k,g_{\to\alpha}(x_0),\ldots,g_{\to\alpha}(x_k)).
\]
We have to show that this set spans a simplex in $Y_\alpha$.
Since $Y_\alpha$ is a flag complex, it suffices to show that all edges
formed by two (distinct) vertices are in $Y_\alpha$.
Since $\sigma\in Y_\alpha$, every pair $(x_i,x_j)$ is in $Y_\alpha$ as well.
Since $g_{\to\alpha}$ is simplicial, the same is true for pairs
of the form $(g_{\to\alpha}(x_i),g_{\to\alpha}(x_j))$.
For pairs of the form $(x_i,g_{\to\alpha}(x_i))$, note that $g_{\to\alpha}(x_i)$ contains $x_i$, so the edge is also contained in the nerve.
Finally, for pairs $(x_i,g_{\to\alpha}(x_j))$ with $i\neq j$, we observe
that the pixels $x_i$ and $x_j$ are intersecting and since $g_{\to\alpha}(x_j)$
contains $x_j$, the pair is also intersecting.
\end{proof}

This lemma shows that for all $\alpha$, $g_{\to\alpha}^\ast$ 
is an invertible map and hence an isomorphism.
Therefore, the filtrations $(Y_\alpha)_{\alpha\geq 0}$ and
$(\complex_\alpha)_{\alpha\geq 0}$ are persistence equivalent,
finishing the proof of Lemma~\ref{lemma:pixel_nerve_iso}.

\subsection{Algorithm to compute the tower}
\label{subsection:appendix-section3-algo}

First, we describe the flooding algorithm to find the pixels in the $\alpha$-ball
of each point $p\in P$.
The complex $\complex_{\alpha}$ is built on the lattice $L_\beta$.
To find the pixels with centers inside $B_1:=B(p,\alpha)$, we first map
$p$ to the nearest lattice point $x$ in $L_\beta$.
Since the diameters of the pixels are upper bounded by $\eps\beta/4$,
a simple triangle inequality shows that for each pixel whose center lies 
in $B_1$, the pixel's center and the centers of its neighboring pixels 
all lie within $B_2:=B(x,\alpha+\eps\beta)$.
It suffices to inspect each pixel in $B_2$ and 
to add it to the list of pixels if the distance of its center 
to $p$ is at most $\alpha$.
We simply enqueue the neighbors of $x$ in a queue $Q$.
We dequeue an element $y\in Q$ and inspect whether the distance to $p$ is
at most $\alpha$; if so, we add it to the list of pixels of $p$.
Also, if the distance to $p$ is at most $\alpha+\eps\beta$, we enqueue the neighbors
of $y$ into $Q$.
We continue dequeuing until $Q$ is empty. 

Since the neighborhood of each pixel $y$  whose center is in $B_1$ is completely
contained inside $B_2$, it is easy to see that there is a sequence of 
pixels $x,c_1,\ldots,c_m,y$ such that consecutive pixels intersect.
In this way, the search finds $y$, 
and terminates after a finite number of steps.

To compute the $k$-skeleton, we store the pixels at $\alpha$ in a dictionary.
We go over each pixel and inspect whether its $3^d-1$ neighbors lie in the
dictionary; if so, we put an edge between the pixel and its neighbor.
The $k$-skeleton can be obtained from the $1$-skeleton by a simple 
combinatorial algorithm,
traversing the Hasse diagram of the complex
(refer to Algorithm 5.8 in~\cite{ckr-polynomial-dcg}).

\subsubsection{Proof of Theorem~\ref{theorem:digital_size}}
\label{subsubsection:appendix-section3-size}

\begin{proof}
First, we bound the computation time required to find the pixels
associated to each point $p\in P$.
We find $x$, the closest point to $p$ in the lattice, and this takes
$O(d)$ time.
To find the pixels, we use a breadth-first search.
A packing argument similar to the one in Theorem~\ref{theorem:digital_size}
shows that the number of vertices of this graph is at most
$\left(\frac{1}{\eps}\right)^d2^{O(d\log d )}$.
Each vertex of this graph has degree $2^{O(d)}$.
Inspecting the distance of a pixel-center to $p$ takes $O(d)$ time.
The breadth-first search then requires 
$\left(\frac{1}{\eps}\right)^d2^{O(d\log d )}$ time.
Since there are $n$ points in $P$, the total time required 
is $n\left(\frac{1}{\eps}\right)^d2^{O(d\log d )}$.

Computing the $k$-skeleton requires $2^{O(dk)}$ time per pixel.
To compute the simplicial map, it suffices to report the vertex map
to the next scale in the tower.
When the lattice does not change between scales, the vertex map
is simply an inclusion, so there is nothing to compute.
Otherwise, we compute the nearest lattice point to each pixel-center
of the current scale and this takes $O(d)$ time per pixel-center.
This cost is subsumed by the cost of the previous operations.
The claim follows.
\end{proof}

\bibliographystyle{plain}

\end{document}